\documentclass{article}

\usepackage{arxiv}

\usepackage[utf8]{inputenc} 
\usepackage[T1]{fontenc}    
\usepackage{hyperref}       
\usepackage{url}            
\usepackage{booktabs}       
\usepackage{amsfonts}       
\usepackage{nicefrac}       
\usepackage{microtype}      
\usepackage{lipsum}

\usepackage{graphicx}
\usepackage{algorithm}
\usepackage[noend]{algorithmic}
\usepackage{bbm}
\usepackage{amsmath, amsthm}
\usepackage{stmaryrd}
\usepackage{amssymb}
\usepackage{paralist}
\usepackage{bm}
\usepackage[numbers]{natbib}
\usepackage[switch]{lineno} 
\usepackage{eurosym}
\usepackage{nccmath} 

\usepackage{chngcntr}
\usepackage{apptools}
\usepackage{caption,subcaption}
\AtAppendix{\counterwithin{lemma}{section}}
\AtAppendix{\counterwithin{theorem}{section}}

\newtheorem{definition}{Definition}
\newtheorem{theorem}{Theorem}
\newtheorem{proposition}{Proposition}

\newtheorem{lemma}{Lemma}
\newtheorem{remark}{Remark}

\newcommand{\GainD}{U^d}
\newcommand{\LossU}{U^u}

\newcommand{\Vset}{\mathcal{V}}
\newcommand{\Cset}{\mathcal{C}}
\newcommand{\Pna}{P_0}
\newcommand{\fna}{f_0}

\newcommand{\pa}{p_a}
\DeclareMathOperator*{\argmax}{arg\,max}
\newcommand{\Gb}{\bm{G}}
\newcommand{\Gbmax}{\Gb_{\max}}

\newcommand{\Gmax}{G_{\max}}

\newcommand{\Ind}{\mathbbm{1}}

\newcommand{\Sset}{\mathcal{S}}

\newcommand\cfa{C_{\text{fa}}}
\newcommand\cd{c_{\textrm{d}}}
\newcommand{\alphaeq}{\alpha^*}
\newcommand{\betaeq}{\beta^*}
\newcommand{\factor}{\ell}
\newcommand{\amount}{A}

\usepackage{color}

\title{Scalable Optimal Classifiers for Adversarial Settings under Uncertainty\thanks{This work was supported by the French National Research Agency through the ``Investissements d'avenir'' program (ANR-15-IDEX-02) and through grant ANR-16-TERC0012; by the DGA; by the Alexander von Humboldt Foundation.}}

\author{
  Benjamin Roussillon\\
  Univ. Grenoble Alpes, Inria, CNRS, Grenoble INP, LIG\\
  \texttt{benjamin.roussillon@inria.fr} \\
   \And
 Patrick Loiseau \\
 Univ. Grenoble Alpes, Inria, CNRS, Grenoble INP, LIG\\
  \texttt{patrick.loiseau@inria.fr} \\
}

\begin{document}
\maketitle

\begin{abstract}
We consider the problem of finding optimal classifiers in an adversarial setting where the class-1 data is generated by an attacker whose objective is not known to the defender---an aspect that is key to realistic applications but has so far been overlooked in the literature. To model this situation, we propose a Bayesian game framework where the defender chooses a classifier with no \emph{a priori} restriction on the set of possible classifiers. The key difficulty in the proposed framework is that the set of possible classifiers is exponential in the set of possible data, which is itself exponential in the number of features used for classification. To counter this, we first show that Bayesian Nash equilibria can be characterized completely via functional threshold classifiers with a small number of parameters. We then show that this low-dimensional characterization enables to develop a training method to compute provably approximately optimal classifiers in a scalable manner; and to develop a learning algorithm for the online setting with low regret (both independent of the dimension of the set of possible data). We illustrate our results through simulations. 
\end{abstract}

\section{Introduction}

Detecting attacks such as spam, malware, or fraud is a key part of security. This task is usually approached as a \emph{binary classification} problem where the defender classifies incoming data (login pattern, text features, or other data depending on the application) as legitimate (non-attack, modeled as class 0) or malicious (attack, modeled as class 1) \cite{Tsai09a,Caruana12a}.


It is well known that using standard classification algorithms for this task leads to poor performance because attackers are able to avoid detection by adjusting the data that they generate while crafting their attacks \cite{Nelson09a,Sommer10a,Thomas13a,Wang14a}. There is a vast literature on adversarial classification (see Section~\ref{sec.related}), but these works often propose ad-hoc defense methods optimized against specific attacks without fully modeling the attacker's adaptiveness. This leads to an arms race between attack and defense papers.

To better take into account the interaction between attacker and defender, several game-theoretic models of adversarial classification have emerged over the last decade (see Section~\ref{sec.related}). 
Most of them, however, have two crucial limitations. First, they restrict the possible classifiers to a specific set of known parameterized classifiers and assume that the defender only selects those parameters. Second, they assume complete information about the attacker's objective,\footnote{with the exception \cite{Grosshans13a}, but which considers regression.} which is often too strong in practice \cite{Vorobeychik18a}.

In a recent paper, Dritsoula et al. \cite{Dritsoula17a} propose a model where the defender can select \emph{any} classifier (i.e., function from the set of data to $\{0, 1\}$). A key difficulty lies in the exponential size of the resulting set of classifiers. The authors show that it is possible to restrict it to a small set of \emph{threshold} classifiers on a function that appears in the attacker's payoff. The classifiers identified, however, have no parameter and their solution method is ad-hoc for the restrictive model chosen---with complete information and simplistic payoffs---, hence it cannot extend to more realistic scenarios. In realistic adversarial classification scenarios with uncertainty on the attacker's payoff, this leaves open the questions: \emph{What classifiers shoud the defender use at equilibrium?} \emph{And how to compute optimal classifiers in a scalable manner?}

In this paper, we answer both questions through the following contributions:
\begin{trivlist}
    \item[1.] We introduce structural extensions to the model of \cite{Dritsoula17a} where the defender can choose any function from a set of data $\Vset$ to $\{0, 1\}$ as a classifier: we model the uncertainty of the defender on the attacker's payoff as a Bayesian game and use generalized payoffs (see Sec.~\ref{sec.model.disc}).
    

    \item[2.] We characterize the equilibrium of the game and exhibit a set of optimal threshold classifiers depending on a small number of parameters (in number independent of $|\Vset|$). Our method first uses a classical technique in resource allocation games (see Sec.~\ref{sec.related}) to establish a link between a mixed strategy on the set of classifiers and a `random classifier' that assigns a probability in $[0,1]$ to every data vector $v\in\Vset$. The set $\Vset$, however, is still exponentially large---this is the key challenge in our work. We then show that the `random classifier' used by the defender at equilibrium has a specific form described with a small number of parameters, and that finding it is equivalent to maximizing a piecewise linear function of the previously mentioned parameters. This low-dimensional characterization has many interesting consequences: it enables using classical stochastic programming and online optimization techniques for efficient learning both online and offline in our game. 
    

    \item[3.]  We show that our parametric expression of equilibrium classifiers allows the defender to train parameters on a labeled dataset with access to only limited information. In particular, our training method, which leverages classical stochastic programming techniques combined with our low-dimensional characterization, produces error bounds independent of $|\Vset|$ and does not require knowledge of the non-attack distribution. This gives much desired scalability since $|\Vset|$ is exponential in the number of features and might be large.

    \item[4.] We illustrate our results through numerical simulations on different games, in particular a credit card fraud game built from the distributions in the publicly available real-world dataset \cite{fraud_set} introduced in \cite{dal2015calibrating}.
    
        \item[5.] We also show that our parametric expression of equilibrium classifiers allows the defender to learn in an online setting---where they update the classifier and receive feedback from the classification at each time step---with very little regret (in particular, independent of $|\Vset|$). 
    
    
\end{trivlist}

Our results provide a basis for designing provably robust classifiers for adversarial classification problems. 
Our characterization of equilibrium strategies also emphasizes the potential of randomized operating point methods from \cite{Lisy14a}: our final classification algorithm can be seen as a randomized operating point on a non-trivial class of optimal threshold classifiers (see discussion below Theorem~\ref{thm.threshold}).
Interestingly, we find that the set of optimal randomized defenses is of low pseudo-dimension. This highlights recent results by Cullina et al.  \cite{pac_adversaries}: in our model, facing adversaries simplifies the learning process as worst-case attacks are predictable while classical learning is chaotic (see discussion below Proposition~\ref{thm.eql}).
This is further supported by our finding that the set of optimal threshold classifiers is of VC dimension $1$ \cite{Shalev-Shwartz14a}; hence our result could be interpreted in hindsight as showing that a reduction to classifiers of VC dimension 1 would come at no loss to the defender. Yet, we emphasize that there is no reason \emph{a priori} why this set would be sufficient, it is a consequence of our results.


\subsection{Related work}
\label{sec.related}

\paragraph{Adversarial learning:}
The literature on adversarial learning usually studies two types of attacks: `poisoning attacks', where the attacker can alter the training set to tamper the classifier's training \cite{Dalvi04a,Globerson06a,Barreno10a,Laskov10a,Huang11a,Zhou14a}; and `evasion attacks', where the attacker tries to reverse engineer a fixed classifier to find a negative instance of minimal cost \cite{Lowd05a,Nelson10a,Li15a}. This literature, however, does not fully model the attacker's adaptiveness, which often leads to an arms race. 
In recent years, the adversarial learning research focused on evasion attacks called adversarial examples that affect deep learning algorithms beyond attack detection applications \cite{Goodfellow15a,Papernot16a,Papernot18a}. These works, however, follow the same pattern. \cite{pac_adversaries} extend PAC theory to adversarial settings and show that fundamental learning bounds can be extended to this setting and that the adversarial VC dimension can be either larger or smaller than the standard one. 

\paragraph{Game-theoretic models of adversarial classification:}
A number of game-theoretic models of adversarial classification have been proposed, with various utility functions and hypotheses on the attacker's capabilities.  
Most of them, however, restrict a priori the possible classifiers: \cite{Zhou12a,Zhou14a} rely on kernel methods; \cite{Kantarcioglu11a} assumes that the defender uses a single type of classifier (though unspecified in the model); \cite{Dalvi04a} focuses on naive Bayes classifiers (and only compute one-stage best responses); \cite{Bruckner11a,Bruckner12a} constrain the classifier to a specific form and look for the (pure) equilibrium value of the parameters; \cite{Li14a} uses a different model but also restrict to linear classifiers; \cite{dasgupta2020improving} restricts the defender to a set of adversarially trained classifiers of different strengths; \cite{Li15a} uses a more general classifier, but restricts for most results to a family of classifiers constructed on a given basis (their model of the attacker is also more constrained than ours); and \cite{Lisy14a} abstracts away the classifier through a ROC curve (attacker and defender only select thresholds). 
In contrast, the objective of our work is to derive the optimal form of the classifiers so we do not make any restriction a priori on the classifiers used. 

At the exception of \cite{Lisy14a, Li15a}, the aforementioned  papers build deterministic classifiers while recent papers tend to advocate for randomization: \cite{bulo2016randomized} introduces random strategies on top of \cite{Bruckner12a} while \cite{Juan19a} highlights the importance of randomized attacks and \cite{pinot2020randomization} of randomized defenses (albeit without being able to characterize the equilibrium). In our work, we completely characterize the equilibrium and naturally find that it must involve randomized attack and defense strategies. 

It is important to understand that these works consider two main types of model. \cite{Bruckner11a, Bruckner12a, pinot2020randomization} study \emph{adversarial learning} problems where the learning problem is defined even without attackers (e.g., image recognition), whereas \cite{Dalvi04a, Lisy14a, Zhou12a, Zhou14a, Kantarcioglu11a, Li14a, Li15a} study \emph{adversarial classification} where the learning problem is to detect attacks and exists only because there are attackers (e.g., spam filtering). These models lead to different attack methods and defenses. Our work belongs to the second category, of adversarial classification problems.

\paragraph{Security games:}
Our game has similarities with \emph{security resource allocation} games \cite{Chen09a,Kiekintveld09a,Bosansky11a,Marecki12a,Fang13a,Balcan15a-short,Schlenker18a, Brown16a} used in applications such as airport security \cite{Pita2009UsingGT}. These works consider a defender with limited resources (e.g. guards) to be allocated to the defense of critical targets. In these settings, problems are at a relatively low scale and are usually entirely described via loss in case of attack of an undefended target. The challenge is the management of the limited amount of resources, which produces NP-hard problems \cite{Korzhyk10a} preventing these models to be transferred to very large scale settings.
Our work studies a similar setting applied to classification, where targets would correspond to attack vectors in $\Vset$. In contrast to the security games literature, we do not impose limited resources (the defender self-restricts its detection to limit false alarm costs), which eliminates the combinatorial issue. We are then able to provide a very different characterization of the solutions with applicability to classification as well as to scale to very large sets $\Vset$ that is never studied in classical security games and is the major challenge in our model.

\paragraph{Exponential zero-sum games:}
Our game reparametrization with `randomized classifiers' to reduce the dimension of the set of classifiers from $2^{|\Vset|}$ to $|\Vset|$ borrows ideas classical in security games.
This technique is also studied for more generic zero-sum games \cite{Immorlica11a}; but with objectives and limitations similar to security games.
%

\section{Model}

In this section we present our game-theoretic model. We introduce utility functions from the defender's viewpoint as we focus on optimal classifiers. We then introduce the probability of detection function as a tool to reduce complexity and discuss the model's assumptions and applicability.

\subsection{Setting and notation}

Consider the following situation. A defender receives data samples that can be either attacks (class 1) or non-attacks (class 0) and wants to predict the class of incoming data. We assume that a data example is represented by a feature vector $v$ that belongs to the same set $\Vset$ regardless of the class.  This vector is typically a simplified representation of the actual attack/non-attack (e.g., spam/non-spam) in a feature space used to perform the classification. We assume that the probability that a data example is an attack, denoted $\pa$, is fixed. 

Vectors corresponding to non-attacks follow a fixed probability distribution $\Pna$ on $\Vset$ whereas vectors corresponding to attacks are generated by attackers. Attackers choose the vector they generate to maximize a utility function (see below) depending on the classification of the defender. To model the uncertainty of the defender, we assume that strategic attackers are endowed with a type $i \in  \llbracket 1, m \rrbracket$ that encodes their utility. The defender does not know the type of the attacker but holds a prior  $(p_i)_{i\in \llbracket 1, m \rrbracket}$ on the possible types. 

The defender chooses a classifier in $\Cset = 2^{\Vset}$, that maps a vector to a predicted class. The defender maximizes a utility function balancing costs/gains in different cases as follows. A \emph{false negative} incurs a loss $\LossU_i(v)$  when facing a type-$i$  attacker. A \emph{true positive} incurs a gain $\GainD_i(v)$  when facing a type-$i$  attacker. A \emph{false positive} incurs a false alarm cost $\cfa(v)$. A \emph{true negative} incurs no cost.
The attacker's gain is the opposite of the defender's for each classification outcome.

Summarizing the above discussion, the utilities of the attacker and defender, when the attacker is of type $i$, are defined as follows: 
\begin{align}
    \label{eq.payoffs}
    U_i^A(v, c) = & \LossU_i(v)\Ind_{c(v) = 0} - \GainD_i(v)\Ind_{c(v) = 1},\\[-1mm]
    U_i^D(v, c) = & - \pa U_i^A(v, c) - (1 - \pa)  \sum_{v' \in \Vset}  \cfa(v^{\prime}) \Pna(v') \Ind_{c(v') = 1}.\nonumber
\end{align}

We assume that $\Vset$ is finite and all functions of $v$ are arbitrary. Our main result, however, extends to $\Vset$ compact (details in Appendix~\ref{sec.continuous} due to space constraints).

The above primitives define a Bayesian game that we denote by $\mathcal{G}$. Note that we assume that all parameters of the game including $\pa$, $P_0$, and the utility functions (but not the attacker's type) are known to both players. (We will discuss later how to relax this assumption.)
As we will see, in this game, equilibria exist only in mixed strategy (intuitively, both players have an incentive to be unpredictable). For the defender, a mixed strategy $\beta$ is a probability distribution on $\Cset$. A mixed strategy of the attacker is a function $\alpha: \llbracket 1, m \rrbracket \to \Delta (\Vset)$ such that for all $i \in \llbracket 1, m \rrbracket$, $\alpha^i_{.}$ is a probability distribution over $\Vset$ chosen by a type-$i$ attacker.
Throughout the paper, we will use the standard solution concept of Bayesian Nash equilibrium, which intuitively prescribes that no player can gain from unilateral deviation.
\begin{definition}
\label{def.BNE}
$(\alpha^*, \beta^*)$ is a Bayesian Nash equilibrium (BNE) of the game $\mathcal{G}$ if and only if, for all $\alpha, \beta$, 
\begin{subequations}
\label{eq.BNE}
\abovedisplayskip=2mm
\begin{align}
\label{eq.BNE-D}
    \sum_{i \in \llbracket 1, m \rrbracket} p_i U^D_i(\alpha^*,\beta^*) &\ge \sum_{i \in \llbracket 1, m \rrbracket} p_i U^D_i(\alpha^*,\beta), \textrm{ and } \\[-1mm]
\label{eq.BNE-A}
    \sum_{i \in \llbracket 1, m \rrbracket} p_i U_i^A(\alpha^*, \beta^*) &\ge \sum_{i \in \llbracket 1, m \rrbracket} p_i U_i^A(\alpha, \beta^*).
\end{align}
\end{subequations}
\end{definition}

The defender's utility depends on the attacker they face. With the belief the defender holds on the probability of each attacker type, it is natural that the defender tries to maximize their average utility. The equilibrium is also described with the average utility of the different attacker types, but as the actions of different attacker types are unrelated it is equivalent to each type maximizing its own utility.

Finally, for all $i \in \llbracket 1, m \rrbracket$, we define $
    \underline{G}_i = \max_{v \in \Vset}\left(-\GainD_i(v) \right) \textrm{ and } 
    \overline{G}_i = \max_{v \in \Vset} \left( \LossU_i(v)\right),
$
which respectively represent the minimum possible gain of the attacker (even if all vectors are always detected they can gain this quantity) and  their maximum possible gain. Note that all results in our paper assume knowledge of these bounds (even when knowledge of utilities is limited). While finding these intervals is challenging if the utilities are arbitrary, they are easy to find in many applications from reasonable monotonicity assumptions on the utilities, as they simply represent the most damage an undetected/detected attack can cause.

\subsection{Preliminary: reduction of dimensionality}

A first difficulty of the model we study is the exponential size of $\Cset$ in $\Vset$. This issue is commonly found in resource allocation games (similar reparametrizations are found in other games such as dueling algorithms) and circumvented through the use of a probability of allocation function: only the probability that an abstract resource is allocated to a target is considered thus ignoring the actual allocation and removing combinatorial complexity (assuming that one can compute this function at equilibrium). In our case, in the spirit of \cite{Dritsoula17a}, we define a probability of detection $\pi$, for any strategy $\beta$ of the defender, as $\pi^{\beta}(v) = \sum_{c\in \Cset} \beta_c \Ind_{c(v) = 1}$.

This transformation exploits the fact that, as long as a vector is detected, the actual classifier used for the detection is not important. Thus, with this probability of detection function, we can rewrite the payoffs independently of classifiers: 
\begin{align}
\nonumber U_i^A(\alpha, \beta) &= \sum_{v \in \Vset}  \alpha^i_{v}\left[\LossU_i(v) -  \pi^{\beta}(v) \cdot \left(\LossU_i(v) + \GainD_i(v)\right)\right];\\[-1mm]
\label{eq.pi} U_i^D(\alpha, \beta) &= - \pa U_i^A(\alpha, \beta) - (1 - \pa) \sum_{v \in \Vset} \cfa(v)\Pna(v)\pi^{\beta}(v).
\end{align}

Any probability of detection function can be attained through simple threshold classifiers crafted for this function.
To see this, consider the set of threshold classifier $c(v) = \Ind_{\pi^{\beta}(v) \ge t}$ for some $t \in[0, 1]$. Then, picking a random threshold uniformly on $[0, 1]$ defines a strategy achieving detection probability $\pi^{\beta}(\cdot)$.

\subsection{Model discussion}
\label{sec.model.disc}

The main motivating scenarios for our model are detection of malicious behaviors such as spam (in emails, social media, etc.), fraud (e.g., bank or click fraud), or illegal intrusion. In such scenarios, the attacker is the spammer, fraudster or intruder while the non-attacker represents a normal user (e.g., non-spam message). The vector $v$ is a representation of the observed behavior on which the classification is done. For spam filtering, it can be a simplified representation of the messages obtained by extracting features such as number of characteristic words. The distribution $\Pna$ represents the distribution over those features for normal messages (not chosen with any adversarial objective). In our basic model, we assume that it is known by both players. It is reasonable in applications where it can be estimated from observation of a large number of easily obtainable messages (e.g., in social medias they are public). We relax it in Section~\ref{sec.stochastic} and Section~\ref{sec.online} where we show that the defender can learn well without a priori knowledge of $\Pna$, $\pa$ and $p_i$.

In our model the defender is uncertain of its own utility as soon as they have uncertainty regarding the attacker they face. Although not the most classical setting, it is meaningful and well studied in Bayesian games (see \cite{Forges92a}). It is well justified in our case. For instance, if a fraudster manages to get access to sensitive information or to an account, the amount of harm may differ depending on the skills and resources of the fraudster. In these fraud settings it makes sense that the attacker's gain is the defender's loss or a fraction of it (e.g. a bank must reimburse its clients or pay higher insurance fees if it is victim of fraud). We note here that our model is still valid for this last case as we rely on zero-sum min-max properties which are robust to small changes such as multiplying factors.

The interaction between classifier and attacker is often modeled as a Stackelberg game where the attacker observes and reacts to the defender's strategy. We focus on the (Bayesian) Nash equilibrium which makes sense if the attacker cannot have perfect information about the defender's strategy. More generally though, we will see that in our game the defender's strategy at BNE must be min-max; hence, any strategy of the defender in a Stackelberg equilibrium would have the same property. We use the Stackelberg model in the online setting where there would be a bigger difference. Note that this min-max property also yields robustness. 
%


Our payoff function generalizes that of \cite{Dritsoula17a} in a practically important way. In their model, a reward $R(v)$ is granted to an attack with vector $v$ regardless of the outcome and a fixed detection cost $\cd$ is paid if the attack is detected. This is unreasonable in many applications such as bank fraud. In our model, the utility in case of detected and undetected attacks are arbitrary unrelated functions of $v$ (which is equivalent to letting the detection cost $c_d$ depend on $v$). This alone breaks the ad-hoc method of \cite{Dritsoula17a} to compute the equilibrium. We also generalize to a Bayesian game (The complete information game is the case where $m=1$), and consider training and online learning problems of practical importance.


\section{BNE characterization and computation}
\label{sec.results}

In this section, we first characterize the equilibrium entirely and exhibit a class of threshold classifiers which are sufficient to define an optimal classifier. Leveraging this characterization, we then show how to compute approximately optimal strategies through training with limited knowledge.

\subsection{Equilibrium characterization}

Finding a Bayesian Nash equilibrium is often hard in general games. 
A key property is that our game is essentially zero-sum and can be reduced to a min-max problem.
Compiling this with the action space reduction via the probability of detection we are able to completely characterize the BNE.

Using the payoffs defined in \eqref{eq.payoffs}, we can see that adding the false alarm term to the payoff of the attacker gives an equivalent Bayesian zero-sum game (as this term is independent from the action of the attacker this addition does not change their strategy). This transformation does not change the defender's payoff. This implies that at equilibrium they maximize their minimum average gain and gives the following lemma (whose proof can be found in Appendix~\ref{proof.lem.minmax}):
\begin{lemma}
\label{lem.minmax}
Let $(\alpha^*, \beta^*)$ be a BNE. Then
\begin{equation}
\label{eq.min-max}
\beta^* \in \argmax_{\beta} \min_{\alpha}\sum_i p_i U_i^D(\alpha, \beta).
\end{equation}
\end{lemma}

Computing the min-max strategies of Lemma~\ref{lem.minmax} can be done via a classical transformation to a linear program, but this ``naive'' program would be of size exponential in $|\Vset|$. Even by expressing it in terms of $\pi^{\beta}$, the program would remain of size $|\Vset|$, which may be too large.
Instead, we will leverage the min-max property to show that the equilibrium can be described compactly using a small number of parameters $\Gb = (G_1, \cdots, G_m)$ that can be interpreted as the utility of the attacker for each type.
Formally, we define:
\begin{definition}[Optimal probability of detection]
\label{proba.optim}
For any $\Gb \in [\underline{G}_1, \overline{G}_1]\times ... \times[\underline{G}_m, \overline{G}_m]$, let 
\abovedisplayskip=1mm
\begin{equation}
\label{eq.piG}
\pi_{\Gb}(v) = \max \left\{0, \max_i\left\{\frac{\LossU_i(v) - G_i}{\LossU_i(v) + \GainD_i(v)}\right\}\right\}, \; \forall v\in \Vset. 
\end{equation}
\end{definition}
As we will see, this quantity is the unique probability of detection that guarantees attacker utility below $\Gb$ while minimizing the false alarms, so it plays a key role in the BNE strategy.
In particular, it allows us to express the strategy of the defender as the maximum of a concave function of $\Gb$: 
\begin{definition}[Minimum gain function $U^D$]
\label{def.mingain}
For all $\Gb \in [\underline{G}_1, \overline{G}_1]\times ... \times[\underline{G}_m, \overline{G}_m]$, let $
U^D(\Gb) = -\pa\sum_i p_i G_i - (1 - \pa)\sum_{v \in \Vset} \cfa(v)\Pna(v) \pi_{\Gb}(v).$
\end{definition}

This function represents the minimum utility of the defender assuming they use a probability of detection function $\pi_{\Gb}(\cdot)$ for some $\Gb$. %
It allows us to state our parametrization result which is the main tool we use to prove all our core results.
\begin{proposition}
\label{thm.eql}
For any 
    $\Gbmax \in \argmax_{\Gb \in [\underline{G}_1, \overline{G}_1]\times ... \times[\underline{G}_m, \overline{G}_m]}(U^D(\Gb))$,
 any strategy of the defender that yields a probability of detection function $\pi_{\Gbmax}(v)$ for all $v \in \Vset$ is a min-max strategy and $\max_{\beta} \min_{\alpha}\sum_i p_i U_i^D(\alpha, \beta) = U^D(\Gbmax)$.
\end{proposition}

A proof of Proposition~\ref{thm.eql} can be found in Appendix~\ref{proof.thm.eql}. The proof relies on the min-max property of the problem which implies that the defender must maximize their minimum gain. We show that for a given utility profile $\Gb$, the minimum gain of the defender as defined in \eqref{eq.min-max} is at least  $U^D(\Gb)$. 
However, the key difficulty is that not all utility profiles $\Gb \in [\underline{G}_1, \overline{G}_1]\times ... \times[\underline{G}_m, \overline{G}_m]$ are feasible and the set of feasible utility profiles needs not be convex due to our Bayesian game and arbitrary functions; hence $U^D(\Gb)$ could be meaningless. Our proof bypasses this difficulty by showing that $\pi_{\Gb_{\max}}(\cdot)$ is a min-max strategy in any case and shows as a corollary that $\Gb_{\max}$ is a feasible utility profile.

Proposition~\ref{thm.eql} states that in order to find the equilibrium strategy, the defender should only find $m$ parameters ($G_1, \cdots, G_m$), corresponding to the maximum utility that it should let each attacker type gain. From those parameters, the probability of detection function is naturally defined. This has multiple consequences.

First, from this characterization we deduce that one does not need to know all the parameters of the problem to find a good strategy. Finding ``good enough'' parameters for the utility of the different attacker types allows the defender to fully define its strategy. This is the main tool allowing us to define strategies which can generalize to unknown vectors in Section~\ref{sec.stochastic}.
In particular, in Theorem~\ref{thm.exponent} we prove that near-optimal (and even optimal with high probability) classifiers can be computed by training the model on a labeled dataset with very limited information. 
Note that this is a key difference between our work and security games where the probability of allocation is computed directly using a linear program. There, the lack of a simple expression for the allocation probability prevents the definition of strategies that can generalize. It is also worth noting that unlike linear programs, our method can be generalized to a continuous vector set---we refer to Appendix~\ref{sec.continuous} for details about that.


Second, the result from Proposition~\ref{thm.eql} shows that the presence of strategic adversaries \emph{simplifies} learning in our problem. Indeed, the class of real valued functions $\{\pi_{\Gb}\}$ which contains the optimal strategy is of low pseudo-dimension (e.g., if there exist $v_1$ (resp. $v_0$) of class 1 (resp $0$) with $\LossU(v_0) > \LossU(v_1)$ and $\GainD(v_0) < \GainD(v_1)$, these two points cannot be shattered). 
This can be explained by the predictable aspect of adversaries acting according to their best-response. On the contrary, when facing non-strategic adversaries the optimal strategy would be a cost-sensitive adaptation of the naive Bayes classifier, which can potentially be any arbitrary function of $2^{\Vset}$ (since we make no assumption on $\Pna$). 
This is noteworthy as such a possibility was hinted at by Cullina et al. \cite{pac_adversaries} who show that, for adversaries who can modify vectors in some neighborhood, the adversarial VC dimension can be either lower or higher than the standard one---i.e., the complexity can either increase or decrease in the presence of adversaries. In our adversarial classification model, the complexity drastically decreases. This suggests that classifiers relying on simply adapting classical training might be inefficient as they do not take into account the fundamental complexity differences between classical and adversarial learning.

With Proposition~\ref{thm.eql} describing the probability of detection function at equilibrium, we can deduce a characterization in terms of threshold classifiers. 
\begin{definition}[Generalized threshold classifiers]
For all $\Gb \in \mathbb{R}^m$, define
\begin{align*}
\Cset^T_{\Gb} = \{c \in \Cset: c(v) = \Ind_{\pi_{\Gb}(v) \ge t},
\forall v \in \Vset \textrm{ for some } t \in [0, 1]\}.
\end{align*}
\end{definition}
\begin{theorem}
\label{thm.threshold}
There exists $\Gb  \in  \mathbb{R}^m$ such that the defender can achieve equilibrium payoff using only classifiers from $\Cset^T_{\Gb}$.
\end{theorem}

This theorem settles our first main question: ``which classifiers should the defender use at the equilibrium?''. These are threshold classifiers on a non-standard function with threshold $t$ representing a probability of detection. A threshold $t$ can be interpreted as classifying as attack if, even when being detected with probability $t$, at least one type of attacker gains at least $G_i$ on average. 
Interestingly, $\Cset^T_{\Gb}$ has a VC dimension of only $1$ as the set comprised of $v_1$ (resp. $v_0$) of class $1$ (resp. $0$) with $\pi_{\Gb}(v_1) < \pi_{\Gb}(v_0)$ cannot be shattered. This strengthens our previous remark on the complexity of adversarial classification. Efficient randomized classification for adversarial settings does not require high capacity classifiers but rather classifiers tailored to the players payoffs. Then, our threshold classifiers may be linear classifier if payoffs are linear as the condition $\pi_{\Gb}(v) \ge t$ can be rewritten as $\max_i\left\{\LossU_i(v) - G_i - t (\LossU_i(v) + \GainD_i(v)) \right\} \ge 0$. Thus, in the linear setting, our threshold classifiers correspond to the defender picking a linear classifier for each type of attacker and outputting class $1$ if at least one of the linear classifiers outputs it. In general however, linear classifiers may perform sub optimally.


The fact that the defender uses specifically threshold classifiers is noteworthy as there is already a literature on the choice of threshold and on this choice in an adversarial setting as in \cite{Lisy14a}. However, the random choice of the threshold in our setting is surprisingly simple -- it is a threshold on the probability of detection and choosing a threshold uniformly over $[0, 1]$ gives the desired strategy. This emphasizes that randomization is necessary to defend against an adversary but also that the choice of the set of classifiers to use is crucial to obtain good results.

Having characterized the equilibrium, we must now answer our second main question ``How can the defender compute optimal strategies in a scalable manner?''. Before presenting a scalable training procedure exploiting our equilibrium parametrization to compute an approximate equilibrium (Section~\ref{sec.stochastic}), let us notice that the equilibrium characterization naively leads to a linear programming solution polynomial in $|\Vset|$ to compute an exact equilibrium as function $U^D$ is piecewise linear. This is presented in Proposition~\ref{lem.incomplete.defender}; note that a similar program could be obtained without our equilibrium characterization. We give in Appendix~\ref{sec.attackerstrategy} a linear program that allows computing the attacker's strategy in time polynomial in $|\Vset|$.

%
\begin{proposition}
\label{lem.incomplete.defender}
Maximizing $U^D(\Gb)$ is equivalent to solving the linear program:
\useshortskip
\begin{equation*}
\begin{array}{ll}
\underset{\pi, \Gb}{\text{maximize}} & -\pa \displaystyle\sum\limits_{i=1}^{m} p_i G_i - (1 - \pa)\sum_{v \in \Vset} \cfa(v)\Pna(v) \pi_v\\
\text{subject to:} & G_i \ge \LossU_i(v) - \pi_v (\LossU_i(v) + \GainD_i(v)), \forall i, \forall v\\
&\pi_v \le 1, \forall v.
\end{array}
\end{equation*}
\end{proposition}

\subsection{Scalable approximate computation}
\label{sec.stochastic}

Our previous results allow computing the equilibrium in time polynomial in $|\Vset|$. Yet, two major challenges remain: ($i$) $|\Vset|$ may be too large, in particular it grows exponentially with the number of features $k$; and ($ii$) computing the equilibrium requires knowledge of all parameters of the game and in particular of $\Pna$, which can be hard to evaluate. In this section, we propose a training method that solves both issues by leveraging stochastic programming techniques. To do so, we first express $U^D(\Gb)$ as an expected value as follows: $U^D(\Gb) = E[U^D(\Gb, \xi)]$ where $U^D(\Gb, \xi) = G_i$ with probability $\pa p_i$ and $U^D(\Gb, \xi) = \cfa(v)\pi_{\Gb}(v)$ with probability $(1 - \pa)\Pna(v)$ for all $v \in \Vset$. Leveraging the specific form of this stochastic function, we apply a stochastic programming technique called sample average approximation (SAA) \cite{shapiro2003monte,kim2015guide,verweij2003sample,linderoth2006empirical} to obtain our training method, Algorithm~\ref{algo.stoc}.

 \begin{algorithm}
\caption{Sample average approximation}
\begin{algorithmic}
\label{algo.stoc}
\STATE Sample $\xi_1 \dots, \xi_N$
\STATE Define $\tilde{U}^D(\Gb) = 1/N \sum_{i = 1}^N U^D(\Gb, \xi_i)$
\STATE Maximize $\tilde{U}^D(\Gb)$ on $[\underline{G}_1, \overline{G}_1]\times ... \times[\underline{G}_m, \overline{G}_m]$
\end{algorithmic}
\end{algorithm}

The maximization step in Algorithm~\ref{algo.stoc} can be done exactly through a linear program in the spirit of Proposition~\ref{proba.optim}, in time polynomial in $N$ since $\tilde{U}^D(\Gb)$ is piecewise linear. Thus the complexity of this algorithm depends only on the sample size and not on the problem dimension. Additionally, very little information is required: the defender only needs to have access to $N$ samples, which may correspond to a labeled dataset, as well as to the parameters $\cfa(v)$, $\LossU_i(v)$, $\GainD_i(v)$ for those samples, and $\underline{G}_i$, $\overline{G}_i $. Yet the following theorem shows that Algorithm~\ref{algo.stoc} outputs an very good approximation of the defender's min-max strategy.
\begin{theorem}
\label{thm.exponent}
Let $\hat{\mathcal{S}}$ be the set of maximizers of $\tilde{U}^D(\Gb)$ from Algorithm~\ref{algo.stoc} and $p_N = Pr [\hat{\mathcal{S}} \subseteq \argmax U^D(\Gb)]$. We have \useshortskip
\begin{equation*}
    \limsup_{N \rightarrow \infty} \frac{1}{N} \log(1 - p_N) < 0.
\end{equation*}
\end{theorem}

A proof of Theorem~\ref{thm.exponent} can be found in Appendix~\ref{proof.thm.training}. It relies on a strong result for sample average approximation (Theorem~$15$ of \cite{shapiro2003monte}), which fully exploits the structure of our problem as it requires the optimized stochastic function to be piecewise linear and to depend on random variables with finite support (extensions to continuous supports are possible under mild assumptions). This result is then enabled by the polyhedral structure of the problem.

Theorem~\ref{thm.exponent} states that Algorithm~\ref{algo.stoc} will find an exact maximum of $U^D(\Gb)$ with probability exponentially close to one (where the randomness is in the draw of the training set from unknown $\Pna$, $\pa$ and $p_i$). Then, from Theorem~\ref{thm.eql}, this immediately gives an exact min-max strategy of the defender. The rate of the exponential convergence of $p_N$ to $1$ is not given by Theorem~\ref{thm.exponent}. It is possible to state a stronger result that gives the rate if the problem is ``well conditionned''---which roughly means that $\argmax U^D(\Gb)$ is a singleton and the function is not flat around the optimum---, but this is not guaranteed in any instance of our game, and such a result is anyways impractical because it depends on the true optimal value. 
From the high-probability result of Theorem~\ref{thm.exponent}, it is easy to derive that the output of Algorithm~\ref{algo.stoc} is exponentially close to the true optimum since the function is bounded; although the exponential rate may be arbitrarily low if the problem is not well conditioned. In that case, though, worst case bounds show convergence of expected value at least in $N^{-1/2}$ and depending only on $Var [U^D(\Gbmax, \xi)]$ \cite{shapiro2003monte}.


Theorem~\ref{thm.exponent} combined with Theorem~\ref{thm.eql} shows that using SAA on top of our equilibrium characterization solves the key difficulties of our problem: we are able to compute an exact min-max strategy for the defender with high probability from a labeled training set without knowledge of $\Pna$, $\pa$ and $p_i$. It is remarkable that we do not need to estimate $\Pna$ from the training set, this is automatically done within the stochastic approximation procedure. Other stochastic approximation algorithms (e.g., as stochastic gradient descent) could be used but without strong convexity property (which is our case since our function is piecewise linear), they only have convergence guarantees in $N^{-1/2}$.



\subsection{Numerical illustration}
\label{sec.experiments}

We performed numerical experiments with different games to illustrate various aspects of our results. In particular, we performed experiments on controlled artificial setups to illustrate the convergence of our training method, the (in)dependence on the number of features, and the form of the equilibrium with multiple attacker types. Due to space constraints, the results are deferred to Appendix~\ref{app.experiments}, along with details on the experimental setup for reproducibility (all our code will be made public upon acceptance). We present here the results for a  game defined with a real feature distribution from a credit card fraud dataset \cite{fraud_set}, to illustrate the form of the equilibrium for simple payoffs.

The dataset \cite{fraud_set} contains transactions made by European cardholders in September 2013. A data vector is composed of $31$ features: the amount of the transaction (in \euro{}) denoted $\amount$, the time since the first transaction in the dataset, whether the transaction was malicious (i.e., the label), and 28 anonymized features coming from a PCA. We instantiate our game with this static data set by replacing each attack in the data set by an abstract adaptative attack in our model. For simplicity, we focus only on the amount of the transaction and consider a single attacker type with the following gains: $\LossU(v) = \amount$, $\GainD(v) = 0$, and $\cfa(v) = \factor \times \amount$ for a given $\factor>0$. This models an attacker that gains the transaction's amount if successful (and the bank loses it), but gains nothing if detected. On the other hand, when a valid transaction is blocked, the bank pays a fraction $\factor$ of the transaction as false alarm cost. This choice of utility functions is meant to illustrate the equilibrium in a reasonable and simple scenario and not to represent a practical ready-to-implement setting. In the dataset, the fraction of attacks is $\pa = 0.00172$, the maximum transaction is $25,691.16$\euro{} with an average of $88.35$\euro{}. There are $N=284,807$ transactions in total.

Figure~\ref{empirical} represents the histogram of valid transaction amounts in $[0, 700]$ (where the majority of transactions occur) and the probability of detection function $\pi_G$ obtained through our training for different values of $\ell$ ($G_{\ell}$ denotes the parameter trained on the dataset with false alarm cost factor $\ell$). When $\ell$ is small, the defender classifies ``aggressively'' as fraud by accepting a  high false alarm rate. When $\ell$ increases, the probability of detection functions show that the defender flags as fraud less often. For example, transactions of $700$\euro{} are flagged with probability $\sim 0.9$ by the most aggressive strategy ($\ell = 0.006$) but only with probability $\sim 0.1$ for the least aggressive strategy ($\ell = 0.074$).

\begin{figure}
\centering
\includegraphics[width = 0.4\columnwidth]{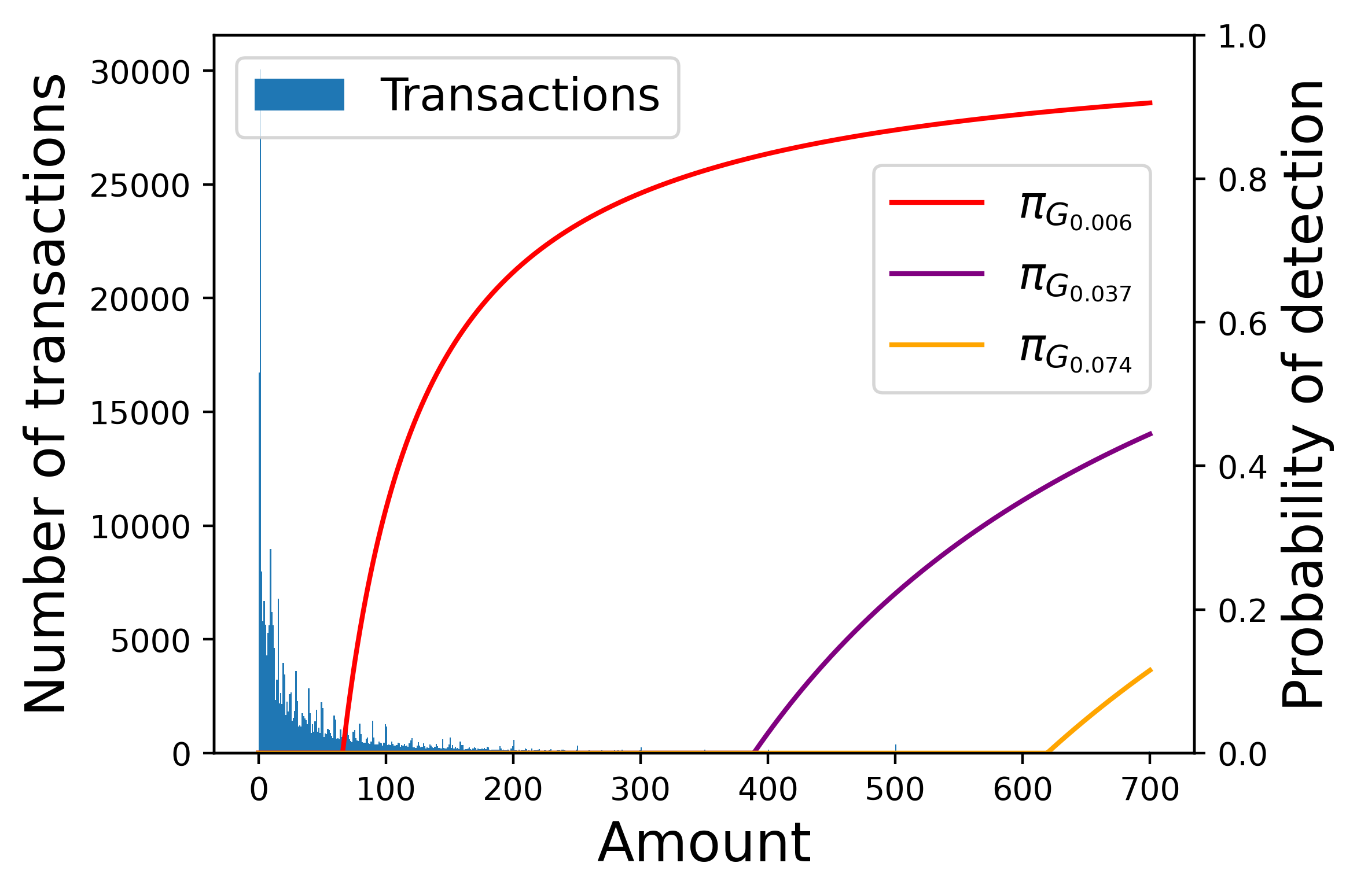}
\caption{Empirical distribution of transaction amounts and representation of defender min-max strategies for various $l$.}
\label{empirical}
\end{figure}
The results presented here are computed through our training method in Algorithm~\ref{algo.stoc} and may not be exact. We evaluate the quality of our approximation on games based on artificial distributions (as we only have access to the empirical distribution). The results suggest that the approximation is good even for much smaller training sets as hinted by the theoretical guarantee. Computation times (not exceeding $15$min) can be found in Appendix~\ref{app.experiments}.


\section{Online learning}
\label{sec.online}

In the previous section, we showed how the defender can compute an approximate min-max strategy from a training set. Yet, such historical data is not always available. We now show how our low-dimensional characterization of the min-max strategy also allows the defender to learn a good strategy \emph{on-line}, without a priori knowledge of $\Pna$, $p_a$ and $p_i$, while incurring low loss as captured by the regret. 

We consider the following setting. At each time step $t = 1, \dots, T$, the defender chooses a probability of detection function $\pi_t$ and receives a vector $v_t$ that is classified as an attack with probability $\pi_t(v_t)$. They incur a loss $l(v_t)$ that is $\cfa(v_t)$ in case of false positive and $0$ in case of true negative if facing a non-attacker; and $- U_i^d(v_t)$ and $U_i^u(v_t)$ in case of true positive and false negative respectively when facing a type $i$ attacker. We assume that after classification, the defender can observe the type of attack (for convenience, we denote by type $i=0$ non-attacks) and that they can compute $\cfa(v_t)$ and $\LossU_i(v_t), \GainD_i(v_t)$ for all $i$. Finally, as in \cite{Chen20a}, we assume that attackers act according to best responses to $\pi_t$ in a Stackelberg fashion, i.e., if the defender faces an attacker of type $i$ at time $t$ we have $v_t \in \arg\max_v\{U^u_i(v) (1 - \pi_t(v)) - U_i^d(v)\pi_t(v)\}$. The defender seeks to minimize the Stackelberg regret:
\begin{definition}[Stackelberg regret]
The Stackelberg regret for a sequence of vectors $(v_1, \dots, v_T)$ is: $ R(T) = \sum_{t=1}^T E_{\pi_t}[ l(v_t)] - \min_{\pi}\sum_{t=1}^T E_{\pi}[l(v_t)].$
\end{definition}

The notion of Stackelberg regret implies that the sequence of vectors depends on the probabilities of detection used. In particular, $\min_{\pi}\sum_t E_{\pi}[l(v_t)]$ must be computed using the best response of the attacker to $\pi$. It is also key to remember that in our setting, the unknown quantities are $\Pna$, $p_a$ and $p_i$. The attacker's strategy is assumed to be known as it is best-response to the utilities $\LossU_i, \GainD_i$. 

It is possible to achieve low regret in $T$ using naively the  online gradient descent algorithm of  \cite{zinkevich2003online}---see Appendix~\ref{app.gradient.descent}---to learn $\pi$ directly. This gives, however, a bound on the Stackelberg regret of  \useshortskip
\begin{equation}
\label{eq.regret-bound}
    R(T) \le \frac{D^2\sqrt{T}}{2} + \left(\sqrt{T} - \frac{1}{2}\right)L^2, 
\end{equation}
with $L = \max(\max_v \{\cfa(v)\}, \max_{v,i}\{|\LossU_i(v) + \GainD_i(v)|\})$ (maximum gradient) and $D^2 = |\Vset|$ (maximum $L_2$ distance between two $\pi$ functions)---see a proof in Appendix~\ref{App.online_naive}. This bound is meaningless if the number of features $k$ is large as $|\Vset| = \Omega (2^k)$. The full strategy $\pi$ also may not fit into memory. 


Building on our characterization of the min-max strategy, we parametrize the defender's strategy by $\Gb$ to propose an alternate learning scheme as Algorithm~\ref{algo.efficient} (where $\Pi_{\mathcal{S}}$ denotes the euclidian projection on a set $\mathcal{S}$). 

\begin{algorithm}
\caption{Efficient online gradient descent}
\begin{algorithmic}
\label{algo.efficient}
\STATE Choose $\Gb_1 \in [\underline{G}_1, \overline{G}_1]\times ... \times[\underline{G}_m, \overline{G}_m]$ arbitrarily
\FOR{$t = 1, \dots, T$}
\STATE Predict $\pi_{\Gb_t}$ and receive vector $v_t$ and type $i$
\IF{$v_t$ came from a non-attacker}
\STATE $\text{grad} \in \partial (\pi_{\Gb_t}(v_t) \cfa(v_t))$
\ELSIF{$v_t$ came from an attacker of type $i$}
\STATE $\text{grad} = e_i$ ($i^{th}$ vector of the canonical base of $\mathbb{R}^m$)
\ENDIF
\STATE $\Gb_{t + 1} = \Pi_{[\underline{G}_1, \overline{G}_1]\times ... \times[\underline{G}_m, \overline{G}_m]}(\Gb_t - \frac{1}{\sqrt{t}} \text{grad})$
\ENDFOR
\end{algorithmic}
\end{algorithm}

Algorithm~\ref{algo.efficient} exploits the fact that each attacker best responds to the defender's strategy, hence only strategies of the form $\pi_{\Gb}(\cdot)$ are worth using. Thus, instead of learning directly $\pi$, the defender learns the parameters $\Gb$. Note that this implies that the defender must be able to evaluate the bounds on the attackers gain they can impose. Algorithm~\ref{algo.efficient} presents two major advantages: 
\emph{First}, the defender's strategy is compactly represented with a small number $m$ of parameters, independent of $|\Vset|$. \emph{Second}, we get a much better regret bound:
\begin{theorem}
\label{thm.online}
Algorithm~\ref{algo.efficient} gives Stackelberg regret bound \eqref{eq.regret-bound} with\\
 $L = \max\{1, \max_{v, i} \{\frac{\cfa(v)}{\GainD_i(v) + \LossU_i(v)}\}\}$ and $D = ||\overline{\Gb} - \underline{\Gb}||_2$. 
\end{theorem}

Theorem~\ref{thm.online} is proved in Appendix~\ref{proof.online}; the proof leverages our characterization of the min-max strategy with parameters $\Gb$. The result formalizes the intuition that learning $\Gb$ rather than $\pi$ allows a much smaller regret ($D$ is now independent in $|\Vset|$). Parameter $L^2$ now represents the change in false alarm cost one can expect at worst when changing parameters $\Gb$; which is different from $L^2$ in the naive procedure that corresponded to a gradient wrt $\pi$. 
We performed numerical experiments that illustrate the result of Theorem~\ref{thm.online} (in particular the independence in $|\Vset|$) in Appendix~\ref{app.experiments}. In addition, we observe that $\Gb_t$ converges towards $\Gbmax$.


\section{Concluding remarks}

We provided a low-dimensional characterization of the min-max strategy in adversarial classification games with general payoffs and showed that this characterization enables efficient training and online learning in practice. Our characterization also allows extending our results to continuous (compact) sets of data $\Vset$---see the details in Appendix~\ref{sec.continuous}.

We considered here only strategic attackers. It is possible to extend the model to include non-strategic attackers that follow fixed strategy, through a redefinition of the false alarm cost that preserves the game structure and allows all our results to be transferred. This can model attacks that are the result of a fixed algorithm. Attacks that are the result of an adaptive algorithm are outside the scope of the current work, but we note that for a wide class of adaptive algorithm this may be modeled in the long run through a utility function. 


\newpage

\bibliography{references-all,references}
\bibliographystyle{plain}

\newpage

\fontsize{10pt}{11pt} \selectfont
\appendix

\onecolumn

\begin{center}
    \Large \bf
    Appendix
\end{center}

\section{Omitted proofs}

In this section, we provide the proofs of all claims in the body of the paper and make a few additional remarks on the proofs.

\subsection{Proof of Lemma~\ref{lem.minmax}}
\label{proof.lem.minmax}

\begin{proof}[Proof of Lemma~\ref{lem.minmax}]
Let $(\alpha^*, \beta^*)$ be a BNE. From the definition of a BNE (see \eqref{eq.BNE-A}), we have $\alpha^* \in \argmax_{\alpha} \sum_i p_i U_i^A(\alpha, \beta^*)$; that is, $\alpha^*$ is a best response of the attacker to $\beta^*$ (for each attacker type). By observing that the average gain of the defender scaled by a factor $\pa$ which does not change the equilibrium strategies can be written as 
\begin{equation*}
\sum_i p_i U_i^D(\alpha, \beta) = - \sum_i p_i U_i^A(\alpha, \beta) + f(\beta)
\end{equation*} 
where $f(\beta) = \frac{1 - \pa}{\pa}\sum_v \cfa(v)\Pna(v)) \pi^{\beta}(v)$ does not depend on $\alpha$, we deduce that $\alpha^* \in \arg \min_{\alpha} \sum_i p_i U_i^D(\alpha, \beta^*)$. Then from the definition of a BNE again (see \eqref{eq.BNE-D}), we conclude that $\beta^* \in \argmax_{\beta} \min_{\alpha} \sum_i p_i U_i^D(\alpha, \beta)$.
\end{proof}

\begin{remark}
Note that by symmetry we also have $\alpha^* \in \arg \max_{\alpha}\min_{\beta} - \sum_i p_i U_i^D(\alpha, \beta)$. Hence, for any $\alpha^*$ such that  $\alpha^* \in \arg \max_{\alpha}\min_{\beta} - \sum_i p_i U_i^D(\alpha, \beta)$ and $\beta^*$ such that $\beta^* \in \argmax_{\beta} \min_{\alpha} \sum_i p_i U_i^D(\alpha, \beta)$, $(\alpha^*, \beta^*)$ is a BNE.
\end{remark}

\subsection{Proof of Theorem~\ref{thm.eql}}
\label{proof.thm.eql}

\begin{proof}[Proof of Theorem~\ref{thm.eql}]
In this proof, to simplify the exposition, we assimilate $\beta$ and $\pi^{\beta}$ and write by abuse of notation $U^A_i(\alpha, \pi)$ and $U^D_i(\alpha, \pi)$ to denote the attacker and defender payoff \eqref{eq.pi} for any defender strategy $\beta$ such that $\pi^{\beta} = \pi$.

Let $\Gbmax \in \argmax_{\Gb} U^D(\Gb)$ and let $\pi_{\Gbmax}(.)$ be the associated probability of detection function \eqref{eq.piG}. We show that $\pi_{\Gbmax}(.)$ is a min-max strategy in two steps.

\begin{itemize}
\item[\underline{Step 1:}] Let $\Gb$ be any arbitrary vector in $[\underline{G}_1, \overline{G}_1]\times \cdots \times[\underline{G}_m, \overline{G}_m]$ and let $\pi_{\Gb}(.)$ be the associated probability of detection function \eqref{eq.piG}. By definition of $\pi_{\Gb}$, we have \begin{equation}
\label{eq.ineqG}
\max_v \{\LossU_i(v) - \pi_{\Gb}(v)\cdot (\LossU_i(v) + \GainD_i(v))\} \le G_i, \quad \forall i \in \llbracket 1, m \rrbracket;
\end{equation}
that is, every type of attacker $i$ can have at most $G_i$ payoff if the defender uses strategy $\pi_{\Gb}$. From the definition of the utility \eqref{eq.pi}, this implies that
\begin{equation*}
\min_\alpha \sum_i p_i U^D_i(\alpha, \pi_{\Gb}) \ge -\pa \sum_i p_i G_i + (1 - \pa)\sum_v \cfa(v) \Pna(v) \pi_{\Gb}(v). 
\end{equation*}
Finally, noting that the rhs of the above inequality is exactly $U^D(\Gb)$ and applying it to $\Gb = \Gbmax$, we obtain
\begin{equation}
\label{eq.ineq1}
    \min_\alpha \sum_i p_i U^D_i(\alpha, \pi_{\Gbmax}) \ge U^D(\Gbmax).
\end{equation}

\item[\underline{Step 2:}] Conversely, let $\pi$ be any arbitrary probability of detection function and define $\Gb^{\pi}$ as the vector with components
\begin{equation}
\label{eq.gain.attacker}
G^{\pi}_i = \max_v\{\LossU_i(v) - \pi(v)\cdot (\LossU_i(v) + \GainD_i(v))\}, \quad (i \in \llbracket 1, m \rrbracket).
\end{equation}
Again, from the definition of the utility \eqref{eq.pi}, we have 
\begin{equation*}
    \min_\alpha \sum_i p_i U^D_i(\alpha, \pi) = -\pa\sum_i p_i G_i^{\pi} - (1 - \pa)\sum_v \cfa(v) \Pna(v) \pi(v);
\end{equation*}
that is that the minimum payoff of the defender is achieved when each attacker type maximizes its gain. Using \eqref{eq.gain.attacker}, we have, for all type $i$ and vector $v$, $\pi(v) \ge \frac{\LossU_i(v) - G_i^{\pi}}{\LossU_i(v) + \GainD_i(v)}$, hence $\pi(v) \ge \max \left\{0, \max_i \left\{ \frac{\LossU_i(v) - G_i^{\pi}}{\LossU_i(v) + \GainD_i(v)} \right\}\right\} = \pi_{\Gb^{\pi}}(v)$ for all $v$. Plugging this inequality in the above equation gives
\begin{equation*}
    \min_\alpha \sum_i p_i U^D_i(\alpha, \pi) \le  -\pa\sum_i p_i G_i^{\pi} - (1 - \pa)\sum_v \cfa(v) \Pna(v) \pi_{\Gb^{\pi}}(v) = U^D(\Gb^{\pi}).
\end{equation*}
Since $U^D(\Gb^{\pi}) \le U^D(\Gbmax)$ for all $\pi$ by definition of $\Gbmax$ as a maximum of function $U^D$, we finally get
\begin{equation}
\label{eq.ineq2}
    \min_\alpha \sum_i p_i U^D_i(\alpha, \pi) \le U^D(\Gbmax).
\end{equation}
\end{itemize}

To conclude, observe that combining \eqref{eq.ineq1} and \eqref{eq.ineq2} gives that, for all $\pi$, $\min_\alpha \sum_i p_i U^D_i(\alpha, \pi_{\Gbmax}) \ge \min_\alpha \sum_i p_i U^D_i(\alpha, \pi)$; hence $\pi_{\Gbmax}$ is a min-max strategy. 
\end{proof}

\begin{remark}
From the proof above, we observe that $\min_{\alpha}(\sum_i p_i U_i^D(\alpha, \pi_{\Gbmax})) = U^D(\Gbmax)$, which implies that, for all $i$ and any maximizer $\Gbmax$ of the function $U^D(\cdot)$, we have for all $i$
\begin{equation*}
    \max_v\{\LossU_i(v) - \pi_{\Gbmax}(v)\cdot (\LossU_i(v) + \GainD_i(v))\} = (\Gmax)_i;
\end{equation*}
that is, when the defender uses strategy $\pi_{\Gbmax}$ the attacker gets a payoff of exactly $(\Gmax)_i$ for all type. 
It is important to note that this is not obvious and it is not the case for all vectors $\Gb$ in $[\underline{G}_1, \overline{G}_1]\times \cdots \times[\underline{G}_m, \overline{G}_m]$. In particular, the set $\mathcal{S} = \{\Gb: \exists \pi, \forall i, G_i = \max_v\{\LossU_i(v) - \pi(v)\cdot (\LossU_i(v) + \GainD_i(v))\}\}$ of all $\Gb$ that are ``best response for each type'' to a strategy $\pi$ is not equal to $[\underline{G}_1, \overline{G}_1]\times \cdots \times[\underline{G}_m, \overline{G}_m]$ and may not even be convex. For a $\Gb$ outside this set $\mathcal{S}$, the maximum payoff of the attacker against strategy $\pi_{\Gb}$ will not be $G_i$ for all $i$, hence the interpretation of $U^D$ as the minimum utility of the defender no longer holds outside $\mathcal{S}$. On the other hand, maximizing $U^D$ on $\mathcal{S}$ directly is not possible as it may not be convex. Our proof bypasses this difficulty by using inequality \eqref{eq.ineqG} that is valid for all $\Gb$ in $[\underline{G}_1, \overline{G}_1]\times \cdots \times[\underline{G}_m, \overline{G}_m]$.
\end{remark}

\subsection{Proof of  Proposition~\ref{lem.incomplete.defender}}
\label{proof.lem.incomplete.defender}

\begin{proof}[Proof of  Proposition~\ref{lem.incomplete.defender}]
Let $OPT$ be the optimal objective value of the linear program.
First note that, for any $\Gb$, the parameters $\pi_{\Gb}$ and $\Gb$ form a valid solution of the linear program by definition. Thus, $\max_{\Gb} U^D(\Gb) \le OPT$.
Conversely, for any optimal solution of the linear program $\pi^*$, $\Gb^*$,  we have $\pi^*_v = \pi_{\Gb^*}(v)$ as it is the probability of detection achieving utility profile $\Gb$ while minimizing false alarms. We also trivially have $G^*_i = \max_v\{\LossU_i(v) - pi_v(\LossU_i(v) + \GainD_i(v))\} \in [\underline{G}_i, \overline{G}_i]$ thanks to the first constraint of the linear program and the fact that we want to minimize the objective function. Thus, $OPT \le \max U^D(\Gb)$.

Combining the two inequalities, we obtain $OPT = \max_{\Gb} U^D(\Gb)$ and $\Gb^* \in \argmax U^D(\Gb)$.
\end{proof}

\subsection{Proof of Theorem~\ref{thm.exponent}}
\label{proof.thm.training}

\begin{proof}[Proof of Theorem~\ref{thm.exponent}]
We start by recalling the setup, assumptions and the main theorem we use from~\cite{shapiro2003monte}.
Consider a stochastic optimization problem of the form
\begin{equation*}
    \min_{x \in X} \{f(x)\} = E[F(x, \xi)],
\end{equation*}
where $\xi$ is a random vector with support $\Xi$, with the following assumptions:
\begin{itemize}
    \item[\textbf{(C1)}] The set $X$ is a convex closed polyhedron;
    \item[\textbf{(C2)}] For every $\xi \in \Xi$ the function $F(\cdot, \xi)$ is proper convex and lower semi continuous and piecewise linear on its domain;
    \item[\textbf{(C3)}] The support $\Xi$ of $\xi$ is finite. 
\end{itemize}
Then the following theorem holds:
\begin{theorem}[\cite{shapiro2003monte}]
Let $\hat{\mathcal{S}}$ be the set of optimal solutions of the sample average approximation problem (Algorithm~\ref{algo.stoc}, with $N$ samples) and $\mathcal{S}$ the set of optimal solution from the true problem. Let $p_N = Pr [\hat{\mathcal{S}} \subseteq \argmax U^D(\Gb)]$. We have \useshortskip
\begin{equation*}
    \limsup_{N \rightarrow \infty} \frac{1}{N} \log(1 - p_N) < 0.
\end{equation*}
\end{theorem}
The result of Theorem~\ref{thm.exponent} then immediately follows by observing that Assumptions \textbf{(C1)}-\textbf{(C3)} trivially hold for the function $U^D(\Gb)$ defined in Definition~\ref{def.mingain} and written as an expectation as explained above Algorithm~\ref{algo.stoc} (with $X = [\underline{G}_1, \overline{G}_1]\times \cdots \times[\underline{G}_m, \overline{G}_m]$ and $\Xi$ is the set of all possible vectors in $\Vset$ and all possible attacker types, hence $|\Xi| = |\Vset| +m$).
\end{proof}

We then justify the fact that bounds in expected value are also relevant for our setting by proving the following lemma: 

\begin{lemma}
\label{lemma.bound.min}
Let $\Gb \in [\underline{G}_1, \overline{G}_1]\times \cdots \times[\underline{G}_m, \overline{G}_m]$. Then, $\min_{\alpha}\sum_i p_i U_i^D(\alpha, \pi_{\Gb}) \ge U^D(\Gb)$.
\end{lemma}
\begin{proof}
By definition of $\pi_{\Gb}$ (see \eqref{eq.piG}), we have for all $i\in \llbracket 1, m \rrbracket$ and $v \in \Vset$, $\frac{\GainD_i(v) - G_i}{\LossU_i(v) + \GainD_i(v)} \le \pi_{\Gb}(v)$. This directly implies $U^A_i(v, \pi_{\Gb}) = \GainD_i(v) - (\LossU_i(v) + \GainD_i(v)) \pi_{\Gb}(v) \le G_i$.
We thus have:
\begin{align*}
\min_{\alpha}\sum_i p_i U_i^D(\alpha, \pi_{\Gb}) &= -\pa\max_{\alpha}\sum_i p_i U_i^A(\alpha, \pi_{\Gb}) - (1 - \pa)\sum_v \cfa(v) \Pna(v) \pi_{\Gb}(v)\\
&\ge - \pa\sum_i p_i G_i - (1 - \pa)\sum_v \cfa(v) \Pna(v) \pi_{\Gb}(v) \\
&\ge U^D(\Gb). \qedhere
\end{align*}
\end{proof}

Lemma~\ref{lemma.bound.min} simply shows that any approximate maximum of $U^D(\Gb)$ also gives an approximate min-max strategy. Thus, any stochastic optimization algorithm which yields bounds in expected value for the minimization of $U^D(\Gb)$ also yields the same expected values guarantees about the minimum gain of the defender.

\subsection{Proof of regret bound for the naive online learning algorithm}
\label{App.online_naive}

Assume that the defender uses Algorithm~\ref{algo.OGDstandard} directly with $\Sset$ being the set of probability distributions on $\Vset$ and with functions $c_t$ such that $c_t(\pi_t) = E_{\pi_t}[l(v_t)]$. If at time $t$ the defender faced an attacker of type $i$, we have $c_t(\pi_t) = \max_{v} [(1 - \pi_t(v)) \LossU_i(v) - \pi_t(v) \GainD_i(v)]$. If the defender faced a non-strategic attacker, we have $c_t(\pi_t) = \cfa(v_t) \pi_t(v_t)$. It is easy to verify that $c_1, \cdots, c_T$ satisfy the conditions of Theorem~\ref{thm.OGD} with $L = \max(\max_v \{\cfa(v)\}, \max_{v,i}\{|\LossU_i(v) + \GainD_i(v)|\})$ and $D^2 = |\Vset|$, hence leading to the regret bound of Theorem~\ref{thm.OGD} with those constants.

\subsection{Proof of Theorem~\ref{thm.online}}
\label{proof.online}

\begin{proof}[Proof of Theorem~\ref{thm.online}]
Algorithm~\ref{algo.efficient} corresponds to online gradient descent from Algorithm~\ref{algo.OGDstandard} applied with $\Sset = [\underline{G}_1, \overline{G}_1]\times \cdots \times[\underline{G}_m, \overline{G}_m]$ and with functions $c_t$ defined as follows:  $c_t(\Gb) = G_i$ if the defender faces an attacker of type $i$ at time $t$ and $c_t(\Gb) = \pi_{\Gb}(v_t) \cfa(v_t)$ if the defender faced a non-attacker at time $t$. 
We first show that 
\begin{equation}
\label{eq.regretmaj}
    \sum_t E_{\pi_{\Gb_t}}[l(v_t)] - \min_{\pi} E_{\pi} [\sum_t l(v_t)]\le \sum_t{c_t(\Gb_t)} - \min_{\Gb} \sum_t c_t(\Gb),
\end{equation} 
in two steps.

\begin{itemize}
\item[\underline{Step 1:}]
Let $\pi^* \in \arg \min_{\pi} E_{\pi} [\sum_t l(v_t)]$ and define
\begin{equation}
G^{\pi^*}_i = \max_v\{\LossU_i(v) - \pi^*(v)\cdot (\LossU_i(v) + \GainD_i(v))\}.
\end{equation}

By definition of $G^{\pi^*}_i$, we have for all $i\in \llbracket 1, m \rrbracket$ and $v \in \Vset$, $\pi^*(v) \ge \frac{\GainD_i(v) - G^{\pi^*}_i}{\LossU_i(v) + \GainD_i(v)}$;  thus, $\pi^*(v) \ge \pi_{\Gb^{\pi^*}}(v)$.

Note that we have $E_{\pi^*}[l(v_t)] \ge c_t(\Gb^{\pi^*})$. Indeed, if at time $t$ a non-attacker was encountered with vector $v_t$, we have $E_{\pi^*}[l(v_t)] = \pi^*(v_t)\cfa(v_t) \ge \pi_{\Gb^{\pi^*}}(v_t)\cfa(v_t) \ge c_t(\Gb^{\pi^*})$; and if an attacker of type $i$ was encountered, we have $E_{\pi^*}[l(v_t)] = \max_v\{\LossU_i(v) - \pi^*(v)\cdot (\LossU_i(v) + \GainD_i(v))\} = G^{\pi^*}_i = c_t(\Gb^{\pi^*})$. 

We thus have $\min_{\pi} E_{\pi} [\sum_t l(v_t)] = E_{\pi^*} [\sum_tl(v_t)] \ge \sum_t c_t(\Gb^{\pi^*}) \ge \min_{\Gb} \sum_t c_t(\Gb)$. 
Additionally, it is trivial to verify that for all $t$ and $\Gb$, we have $E_{\pi_{\Gb}}[l(v_t)] \le c_t(\Gb)$. Combined with the previous inequality, we get: 
\begin{equation}
\label{eq.min.l_t}
    \min_\pi E_{\pi}[\sum_t l(v_t)] = \min_{\Gb}c_t(\Gb).
\end{equation}

\item[\underline{Step 2:}]
As stated above, it is trivial to verify that for all $t$ and $\Gb$, we have $E_{\pi_{\Gb}}[l(v_t)] \le c_t(\Gb)$. In particular, this holds for $\Gb = \Gb_t$, which directly implies:
\begin{equation}
\label{ineq.l_t}
    \sum_t E_{\pi_{\Gb_t}}[l(v_t)] \le \sum_t c_t(\Gb).
\end{equation}
\end{itemize}

Combining \eqref{eq.min.l_t} and \eqref{ineq.l_t} immediately gives~\eqref{eq.regretmaj}.
To conclude the proof, we apply the regret bound of Theorem~\ref{thm.OGD} to the right-hand side of~\eqref{eq.regretmaj}, noting that each loss function is convex and that it can be easily verified from the definitions of $D$ and $L$ that the conditions of Theorem~\ref{thm.OGD} are satisfied. 
\end{proof}


\section{Classical online gradient descent algorithm and associated regret bound}
\label{app.gradient.descent}

In this section, we present the online gradient descent algorithm (termed ``greedy projection'' in \cite{zinkevich2003online}) and the associated regret bound from \cite{zinkevich2003online}. 
Let ($c_1, \dots, c_T$) be convex functions defined on a convex set $\Sset$. Let $\Pi_{\Sset}$ to be the Euclidian projection on $\Sset$, and let $\eta_t$, $t=1, \cdots, T$, be a sequence of learning rates. 
Let $\partial c_t(x)$ denote the set of sub-gradients of $c_t$ at point $x$. The online gradient descent algorithm is as follows.
\begin{algorithm}
\caption{Online gradient descent (OGD)}
\label{algo.OGDstandard}
\begin{algorithmic} 
\STATE Initialize $x_1 \in \Sset$ arbitrarily
\FOR{$t = 1 \dots T$}
\STATE Observe $c_t$
\STATE Select $x_{t + 1} = \Pi_{\Sset}(x_t - \eta_t \nabla c_t(x_t))$ for any $\nabla c_t(x_t) \in \partial c_t(x_t)$
\ENDFOR
\end{algorithmic}
\end{algorithm}

Then, \cite{zinkevich2003online} shows that we have the following regret bound.
\begin{theorem}[Zinkevich, 2003]
\label{thm.OGD}
Assume that 
\begin{equation*}
    \max_{x_1, x_2 \in \Sset}||x_2 - x_1||_2 \le D
\end{equation*} 
and 
\begin{equation*}
    ||\nabla c_t(x)||_2 \le L, \quad \forall t \in \{1, \dots, T\}, x \in \Sset, \forall \nabla c_t(x) \in \partial c_t(x).
\end{equation*} 
Let $(x_1, \dots, x_t)$ be vectors in $\Sset$ selected by Algorithm~\ref{algo.OGDstandard} with $\eta_t = \frac{1}{\sqrt{t}}$. Then, the regret accumulated at time $T$, defined as $R(T) \equiv \sum_{t=1}^T c_t(x_t) - \min_{x} \sum_{t=1}^T c_t(x)$, is bounded by:
\begin{equation*}
    R(T) \le \frac{D^2\sqrt{T}}{2}+ \left(\sqrt{T} - \frac{1}{2}\right)L^2.
\end{equation*}
\end{theorem}

Importantly, note that although the bounds $L$ and $D$ appear in the regret bound of Theorem~\ref{thm.OGD}, it is not necessary to know them to run the online gradient descent algorithm and they are not used in the algorithm. Note also that our functions $c_t$ may not be differentiable. As noted in footnote $3$ of \cite{zinkevich2003online}, the algorithm works also in that case, using sub-gradients as presented above.

\section{Computation of the attacker's equilibrium  strategy}
\label{sec.attackerstrategy}

In this section, we show how to compute the attacker's equilibrium strategy. 

First, as we mentioned in Proposition~\ref{lem.incomplete.defender}, the strategy of the defender can be computed through a linear program. The dual of this linear program, however, does not give the attacker's equilibrium strategy. 
Indeed, the dual is the following:
\begin{equation*}
\begin{array}{ll@{}llll}
\underset{\alpha_v^i, \pi_v}{\text{minimize}} &  \displaystyle\sum_{v \in \Vset} \sum_{i \in \llbracket 1, m \rrbracket} &\alpha^i_v  \LossU_i(v) + \sum_{v \in \Vset} &\pi_v\\
\text{subject to:} & & &\pi_v & \ge 0,& \forall v \in \Vset\\
& & \alpha^i_v & & \le 0, & \forall i \in \llbracket 1, m \rrbracket, \forall v  \in \Vset\\
& \sum_{v \in \Vset} &\alpha^i_v & &= -p_i, &\forall i \in \llbracket 1, m \rrbracket \\
& \sum_{i \in \llbracket 1, m \rrbracket} & \alpha^i_v (\LossU_i(v) + \GainD_i(v)) + & \pi_v &\ge - \frac{1 - \pa}{\pa}\cfa(v)\Pna(v),  & \forall v \in \Vset
\end{array}
\end{equation*}

The second and third constraints could make sense if we  considered the variables $-\alpha_v^i/p_i$, but the first and fourth constraints do not correspond to the problem. Indeed, with these constraints $\pi_v$ is unrestricted so it does not necessarily correspond to a probability of detection function. Additionally, $\alpha^i_v$ does not fit the characterization of the strategy of the attacker at equilibrium given in Lemma~\ref{lem.incomplete-attacker} below.
While it may seem counter-intuitive that the dual of the linear program giving the min-max strategy of the defender does not output the min-max strategy of the attacker, recall that the min-max strategy of the defender was not computed with the standard linear program for min-max problems but through a linear program computing the maximum of a piecewise-linear function.

We now give the following characterization of the attacker's strategy at equilibrium: 
\begin{lemma}
\label{lem.incomplete-attacker}
Let $(\alphaeq, \betaeq)$ be a BNE of $\mathcal{G}$, then:
\begin{align*}
&\forall v \in \Vset \text{ s.t. } 0 < \pi^{\betaeq}(v) < 1:& \sum_i p_i{\alphaeq}^i_v (\LossU_i(v) + \GainD_i(v)) = \frac{1 - \pa}{\pa}\cfa(v)\Pna(v),  \\
&\forall v \in \Vset \text{ s.t. } \pi^{\betaeq}(v) = 0:& \sum_i p_i{\alphaeq}^i_v (\LossU_i(v) + \GainD_i(v)) \le \frac{1 - \pa}{\pa}\cfa(v)\Pna(v), \\
&\forall v \in \Vset \text{ s.t. } \pi^{\betaeq}(v) = 1 \text{ and } v \in \Vset :& \sum_i p_i{\alphaeq}^i_v (\LossU_i(v) + \GainD_i(v)) \ge \frac{1 - \pa}{\pa}\cfa(v)\Pna(v).
\end{align*}
\end{lemma}

\begin{proof}
Similarly to earlier proofs, in this proof, to simplify the exposition we assimilate $\beta$ and $\pi^{\beta}$ and write by abuse of notation $U^D_i(\alpha, \pi^{\beta})$ to denote the defender payoff \eqref{eq.pi}. Let $(\alpha, \beta)$ be a strategy profile; then we have for all $v \in \Vset$:
\begin{equation}
\label{eq.deriv.pi}
    \frac{\partial \sum_i p_i U_i\left(\alpha, \pi^{\beta}\right)}{\partial \pi(v)} = \sum_i {\alpha}^i_vp_i(\LossU_i(v) + \GainD_i(v)) - \frac{1 - \pa}{\pa}\cfa(v)\Pna(v).
\end{equation}

As $(\alphaeq, \betaeq)$ is a BNE, by definition, $\betaeq \in \arg \max_{\beta} \sum_i p_i U_i^D(\alphaeq, \beta)$. This implies that:
\begin{align*}
    &\forall v \in \Vset \text{ s.t: } \pi^{\betaeq}(v) = 0: &  \frac{\partial \sum_i p_i U_i\left(\alphaeq, \pi^{\betaeq}\right)}{\partial \pi(v)} = \sum_i {\alphaeq}^i_v p_i(\LossU_i(v) + \GainD_i(v)) - \frac{1 - \pa}{\pa}\cfa(v)\Pna(v) \le 0,  \\
    &\forall v \in \Vset \text{ s.t: } 0 < \pi^{\betaeq}(v) < 1: &  \frac{\partial \sum_i p_i U_i\left(\alphaeq, \pi^{\betaeq}\right)}{\partial \pi(v)} = \sum_i {\alphaeq}^i_vp_i(\LossU_i(v) + \GainD_i(v)) - \frac{1 - \pa}{\pa}\cfa(v)\Pna(v) = 0, \\
    &\forall v \in \Vset \text{ s.t: } \pi^{\betaeq}(v) = 1: &  \frac{\partial \sum_i p_i U_i\left(\alphaeq, \pi^{\betaeq}\right)}{\partial \pi(v)} = \sum_i {\alphaeq}^i_vp_i(\LossU_i(v) + \GainD_i(v)) - \frac{1 - \pa}{\pa}\cfa(v)\Pna(v) \ge 0;
\end{align*}
which directly concludes the proof.
\end{proof}

Intuitively, Lemma~\ref{lem.incomplete-attacker} states that the attacker's strategy strikes a balance between the risk $\frac{1 - \pa}{\pa}\cfa(v)\Pna(v)$ the defender takes to detect a vector $v$ and the average gain associated with the detection of vectors $\sum_i p_i \alpha_v^i (\LossU_i(v) + \GainD_i(v))$.
It allows us to find the best-response of the attacker $\alphaeq$ to a min-max strategy $\betaeq$ of the defender, hence allowing us to find a BNE.

\begin{proposition}
\label{attacker.linear}
Let $\betaeq \in \arg \max_{\beta} \min_{\alpha} \sum_i p_iU_i^D(\alpha, \beta)$. Then there exists a solution to the linear program
$\mathrm{find}_{\alpha}(\betaeq)$:
\begin{equation*}
\begin{array}{ll@{}l}
\underset{\alpha_v^i}{\text{maximize}} & 0\\
\text{subject to:} & 0 \le &\alpha_v^i \le 1, \forall i \in \llbracket 1, m \rrbracket, \forall v \in \Vset, \\
&\sum_{v \in \Vset} &\alpha_v^i = 1, \forall i \in \llbracket 1, m \rrbracket,\\
& & \alpha_v^i = 0, \forall i \in \llbracket 1, m \rrbracket, \forall v \in \Vset, s.t.: \LossU_i(v) - (\LossU_i(v) + \GainD_i(v)) \pi^{\beta^*}(v) < G_i,\\
& \sum_{i} p_i &\alpha^i_v (\LossU_i(v) + \GainD_i(v)) = \frac{1 - \pa}{\pa}\cfa(v)\Pna(v), \forall v \in \Vset, s.t.:\pi^{\beta^*}(v) \in (0,1),\\
& \sum_{i} p_i &\alpha^i_v (\LossU_i(v) + \GainD_i(v)) \le \frac{1 - \pa}{\pa}\cfa(v)\Pna(v), \forall v \in \Vset, s.t.:\pi^{\beta^*}(v) = 0,\\
& \sum_{i} p_i &\alpha^i_v (\LossU_i(v) + \GainD_i(v)) \ge \frac{1 - \pa}{\pa}\cfa(v)\Pna(v), \forall v \in \Vset, s.t.:\pi^{\beta^*}(v) = 1,\\
\end{array}
\end{equation*}
where $G_i = \max_{v}(\LossU_i(v) - (\LossU_i(v) + \GainD_i(v)) \pi^{\betaeq}(v))$. 
Additionally, for any solution $\alphaeq$ of $\mathrm{find}_{\alpha}(\betaeq)$, $(\alphaeq, \betaeq)$ is a BNE.
\end{proposition}

\begin{proof}
Let $\beta^*$ be a min-max strategy for the defender. Then, from the proof of Lemma~\ref{lem.minmax}, for any $\alpha^* \in \arg \max_{\alpha} \min_{\beta} - \sum_i p_i U_i^D(\alpha, \beta)$, $(\alpha^*, \beta^*)$ is a BNE. Thus, it satisfied the conditions of Lemma~\ref{lem.incomplete-attacker}; using those, it is trivial to check that $\alpha^*$ is a solution of $\mathrm{find}_{\alpha}(\beta^*)$. So this linear program admits a solution.

Next, any solution $\alpha^*$ of $\mathrm{find}_{\alpha}(\beta^*)$ satisfies the conditions of Lemma~\ref{lem.incomplete-attacker}, which implies that
\begin{equation}
\label{BNE.part1}
    \beta^* \in \arg\max_{\beta} \sum_i p_i U_i^D(\alpha^*, \beta).
\end{equation}
Additionally, from the third constraint of the linear program, we observe that for any given $i$, by definition of $G_i$, for all $v \in \Vset$ s.t. $v \notin \arg \max_{v^{\prime}}(\LossU_i(v^{\prime}) - (\LossU_i(v^{\prime}) + \GainD_i(v^{\prime})) \pi^{\betaeq}(v^{\prime}))$, we have $\alpha^{*i}_v = 0$. Thus, we have
\begin{equation}
\label{BNE.part2}
    \alpha^* \in \arg \max_{\alpha} \sum_i p_i U_i^A(\alpha, \beta^*).
\end{equation}

Combining \eqref{BNE.part1} and \eqref{BNE.part2}, we conclude that $(\alpha^*, \beta^*)$ is a BNE.
\end{proof}

The linear program finds a valid strategy of the attacker which fits the conditions stated in Lemma~\ref{lem.incomplete-attacker} with the additional condition that each type of attacker must play his most rewarding vectors. 
While we had to resort to a linear program, finding the strategy of the attacker is still done in polynomial time in $m|\Vset|$ as we have a program with $m|\Vset|$ variables and $\mathcal{O} (m|\Vset|)$ constraints.

\section{Continuous game}
\label{sec.continuous}

In the paper we assumed that the set $\Vset$ of possible data is finite, but we make no other assumption on $\Vset$. That leaves a lot of flexibility; in particular it is possible to model situations where the features are categorical or boolean, or discrete numerical values (or a combination of those). Yet, some features are naturally continuous and it can be convenient to model them as such instead of considering a discretization. One of the advantage of our characterization of the BNE is that is naturally extends to continuous feature spaces, as we sketch next.  

To extend the model to continuous feature space, we assume that $\Vset$ is a compact metric set and that the defender's strategy space $\mathcal{D}$ is the set of $M$-Lipschitz continuous functions $\pi$ from $\Vset$ to $[0, 1]$. We use the same parameters and notation as in the finite game with the exception of $\Pna(\cdot)$, which is now a continuous probability distribution. We assume that it has a density and denote it $\fna$. We then have the following payoffs
\begin{align}
    \label{eq.payoffs.cont}
    U_i^A(v, \pi) = & \LossU_i(v)(1 - \pi(v)) - \GainD_i(v)\pi(v), \nonumber\\ 
    U_i^D(v, \pi) = & - \pa U_i^A(v, \pi) - (1 - \pa ) \int_{\Vset}  \cfa(v) \fna(v) \pi(v^\prime) \mathop{dv^{\prime}}.
\end{align}

Let us suppose that all functions are continuous and integrable. This ensures that the game is continuous and well defined, which, thanks to Glicksberg's theorem \cite{Glicksberg}, ensures the existence of a BNE. 
Finally, we define $\underline{\Gb}$, $\overline{\Gb}$ and $\pi_{\Gb}$ as in the finite case. Let us assume that there exists $M \in \mathbb{R}$ such that, for all $\Gb \in [\underline{G}_1, \overline{G}_1]\times \cdots \times[\underline{G}_m, \overline{G}_m]$ the function $\pi_{\Gb}(\cdot)$ is $M$-Lipschitz . This ensures that the optimal strategies defined in the paper are available to the defender. Note that for this assumption to hold, it is sufficient that there exists $C \in \mathbb{R}$ and $\epsilon \in \mathbb{R}$ such that $\forall v \in \Vset$, we have $\frac{1 - \pa}{\pa}\cfa(v)\Pna(v) \le C$ and $\forall v \in \Vset$, we have  $|\LossU_i(v) + \GainD_i(v)| \ge \epsilon$. The first inequality simply implies that false alarm costs should be bounded and the second one that there should always be at least a small difference between the reward an attacker gets with an undetected attack ($\LossU_i(v)$) and a detected attack ($-\GainD_i(v)$).

Note that in this continuous setting, we defined the game with strategies $\pi$ directly for the defender (defining $\Cset$ for the discrete game was useful to get a finite game but this is irrelevant here) and therefore bypass many technical issues on the potential compactness of $\Cset$. 
We also see that the defender need not use mixed strategies as the payoff of any mixed strategy can be attained with a pure strategy $\pi$ being the average of the functions in the mixed strategy. This is explained by the fact that $\pi$ already represents a random classification, removing the need for further randomization.

We can then extend some of our results to the continuous case, in  particular Theorem~\ref{thm.eql} that leads to the form of optimal classifiers. Again for the continuous case, intuitively, considering any equilibrium with attacker utility profile $\Gb$, the defender must have a probability of detection function $\pi(\cdot) = \pi_{\Gb}(\cdot)$ as it is the probability of detection which gives gain $G_i$ for each attacker while minimizing the false alarm cost. This gives the following proposition.
\begin{proposition}
\label{prop.continuous}
For all $\Gb \in [\underline{G}_1, \overline{G}_1]\times \cdots \times[\underline{G}_m, \overline{G}_m]$, let
\begin{equation*}
U^D(\Gb) = -\pa\sum_{i = 1}^m p_i G_i - (1 - \pa)\int_{\Vset} \cfa(v) \fna(v) \pi_{\Gb}(v)\mathop{dv}.
\end{equation*}
For all $\Gbmax \in \argmax_{\Gb \in [\underline{G}_1, \overline{G}_1]\times \cdots \times[\underline{G}_m, \overline{G}_m]}(U^D(\Gb))$, $\ \pi_{\Gbmax}(\cdot)$ is a min-max strategy.
\end{proposition}

Proposition~\ref{prop.continuous} states that, as in the discrete case, finding a BNE in the continuous case amounts to finding the maximum of a concave function (function $U^D(\cdot)$ defined above). This problem can be solved using classical convex optimization tools, assuming that one can efficiently evaluate the function, in particular the integral that appears in the function. If one is unable to derive an exact formula for the integral or to devise an efficient Monte-Carlo approximation, it is always possible to use discretization of $\Vset$ with sufficient precision, but with the caveat of scalability issues due to the curse of dimensionality. 
Note, however, that our results on online learning remain valid in the continuous setting as the proof of Theorem~\ref{thm.online} does not rely on the assumption that $\Vset$ is finite.

\section{Omitted details and additional results from numerical experiments}
\label{app.experiments}

\subsection{Hardware and software used for experiments}
\label{app.reproducibility}

All experiments were run on a Dell xps-13 laptop with a Quad core Intel Core i7-8550U (-MT-MCP-) CPU under Ubuntu 18.04.
Experiments were made using Python 3 code which is publicly available at \url{https://gitlab.inria.fr/broussil/adversarial-classification-under-uncertainty}.

Random experiments are all seeded and random numbers are generated with numpy.random. They can be reproduced with the Python 3 code.

The main libraries used are ($i$) pulp for linear programming, ($ii$) numpy for algebra, and ($iii$) matplotlib for plots.

\subsection{Additional numerical results}

In this section, we present the results of additional numerical experiments that further illustrate our model and theoretical results and consolidate the observations of Section~\ref{sec.experiments} but did not fit in the body of the paper. To that end, we use three different games with different characteristics. Game 1 is a very simple game with a single numerical feature and with simple and smooth payoffs.  Game 2 is an artificial a bank fraud scenario (similar to the one used for the illustration in the body of the paper, Section~\ref{sec.experiments}). Game 3 is a completely random game with binary features, which is in some sense worst-case as there is no structure to exploit:
\begin{itemize}

\item[\bf Game 1: One-feature game]

Game 1 is a game with two possible types of attacker in which classification is based on a single feature. The attackers' strategy space consists of $101$ attack $v_0, \cdots, v_{100}$. For Attacker 1, an undetected attack yields a utility $\LossU_1(v_r) = r$.  A detected attack incurs a cost $\GainD_1(v_r) = 30*(r \mod(10))$, see  Figure~\ref{fig.parameters}. Attacker 2, whose strategy space is the same as Attacker 1, has $\LossU_2(v_r) = 100 - r$ and bears a cost $\GainD_2(v_r) = 300 - 30 (r \mod(10))$ in case of detected attack. Their gain and cost functions mirror that of Attacker 1, being interested in low vectors while Attacker 1 is interested in high vectors. Hence, the defender faces two attackers with different interests.
There is a proportion $\pa = 0.2$ of attackers. The defender bears a constant false alarm cost $\cfa = 140$.
A non-attacker follows a binomial distribution, they play the vector $v_r$ with probability $\Pna(v_r) = \binom{100}{r}\theta_0^r(1 - \theta_0)^{100 - r}$ with $\theta_0 = 0.2$.

\item[\bf Game 2: Bank fraud game]

Game 2 is a game similar to the one presented in Section~\ref{sec.experiments}, but with features following an artificial (controlled) distribution. We consider vectors $v$ of the form $v = \amount$ where $\amount$ is the amount of the transaction discretized on integer amounts. We define as in Section~\ref{sec.experiments}, $\LossU(v) = \amount$, $\GainD(v) = 0$ and $\cfa(v) = \ell \amount$ for some $\ell>0$. We prescribe the following non-attacker distribution: $\Pna(v) =p(\amount)$ where the amount of the transaction of a user follows a binomial distribution between $0$ and $25691$ with a mean of $88$.

\item[\bf Game 3: Random game]
Game 3 aims to illustrate learning where there is no correlation between vectors and costs as well as the presence of multiple attackers. It is a game with  four types of attackers with vectors of $k$ binary features, for $k$ up to 19 to be able to compute the exact optimum for comparison. For each attacker type $i$ and vector $v$, $\LossU_i(v)$, $\GainD_i(v)$ and $\cfa(v)$ are assigned a random value uniformly between $10$ and $20$. We set $\pa = 0.1$ and generate $p_i$'s randomly. Distribution $\Pna$ is uniform on the $|\Vset|$-dimensional unit simplex.
\end{itemize}

\subsubsection{Illustrations of the BNE}

We first provide here basic illustrations of our results concerning the structure of the BNE. We use Game 1 as it is the easiest to interpret the results. 
Figure~\ref{fig.BNE} illustrates the behavior of both players at NE when Attacker 2 is not present ($p_2 = 0$). The attacker wants to play high vectors but must follow the distribution of the non-attacker over the vector they deem rewarding enough in regards to the defender's strategy to remain stealthy as stated in Lemma~\ref{lem.incomplete-attacker}. 
The defender detects vectors with some spikes in the probability of detection function corresponding to the spikes in the cost incurred by detection. Indeed, at the equilibrium they make the attacker indifferent between some vectors and in order to do so, vectors which suffer from a sudden increase in cost incurred by detection can be detected less.

\begin{figure}
     \centering
     \begin{subfigure}[b]{0.33\textwidth}
         \centering
         \includegraphics[width=\textwidth]{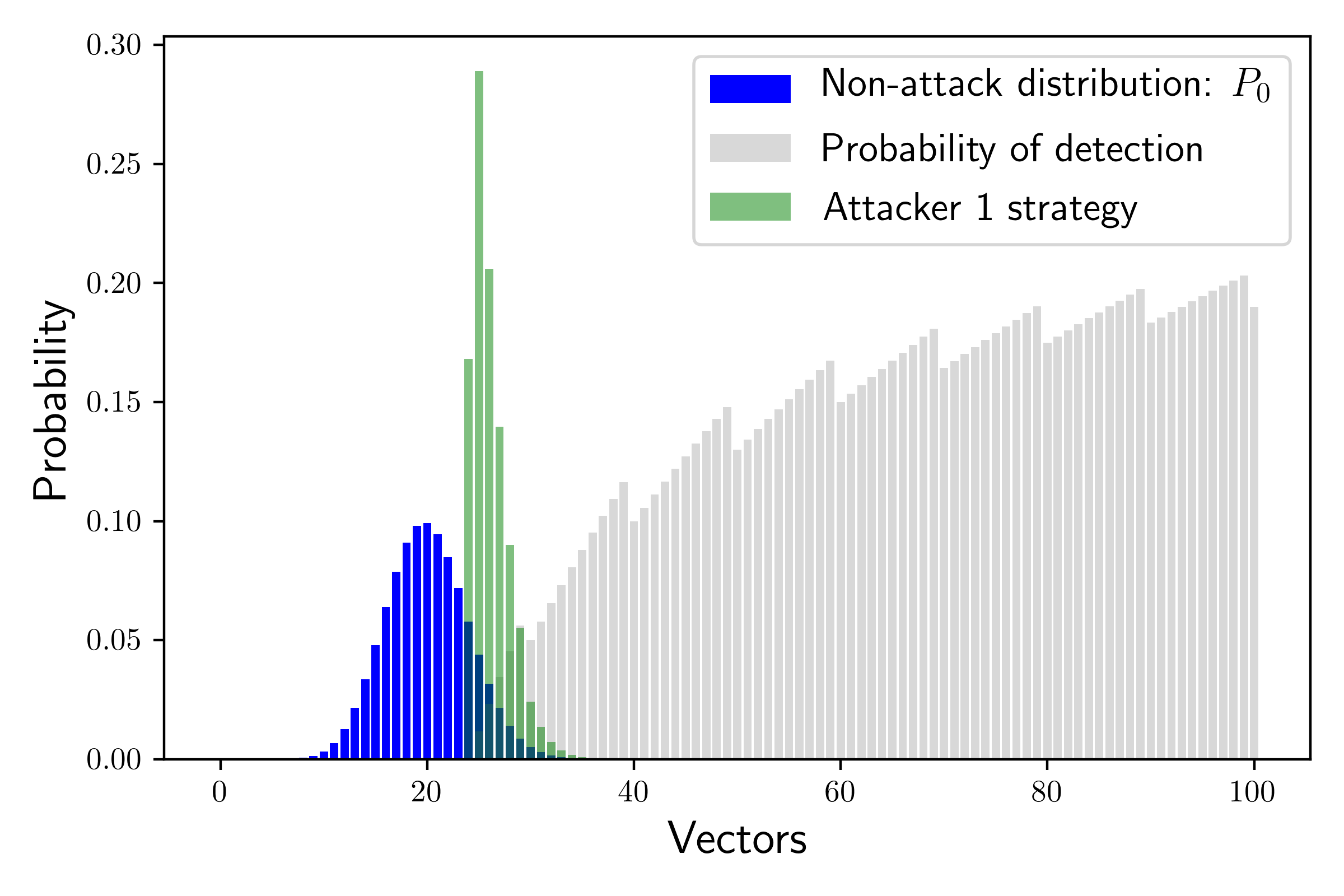}
         \caption{}
         \label{fig.BNE}
     \end{subfigure}
     \begin{subfigure}[b]{0.33\textwidth}
         \centering
         \includegraphics[width=0.9\textwidth]{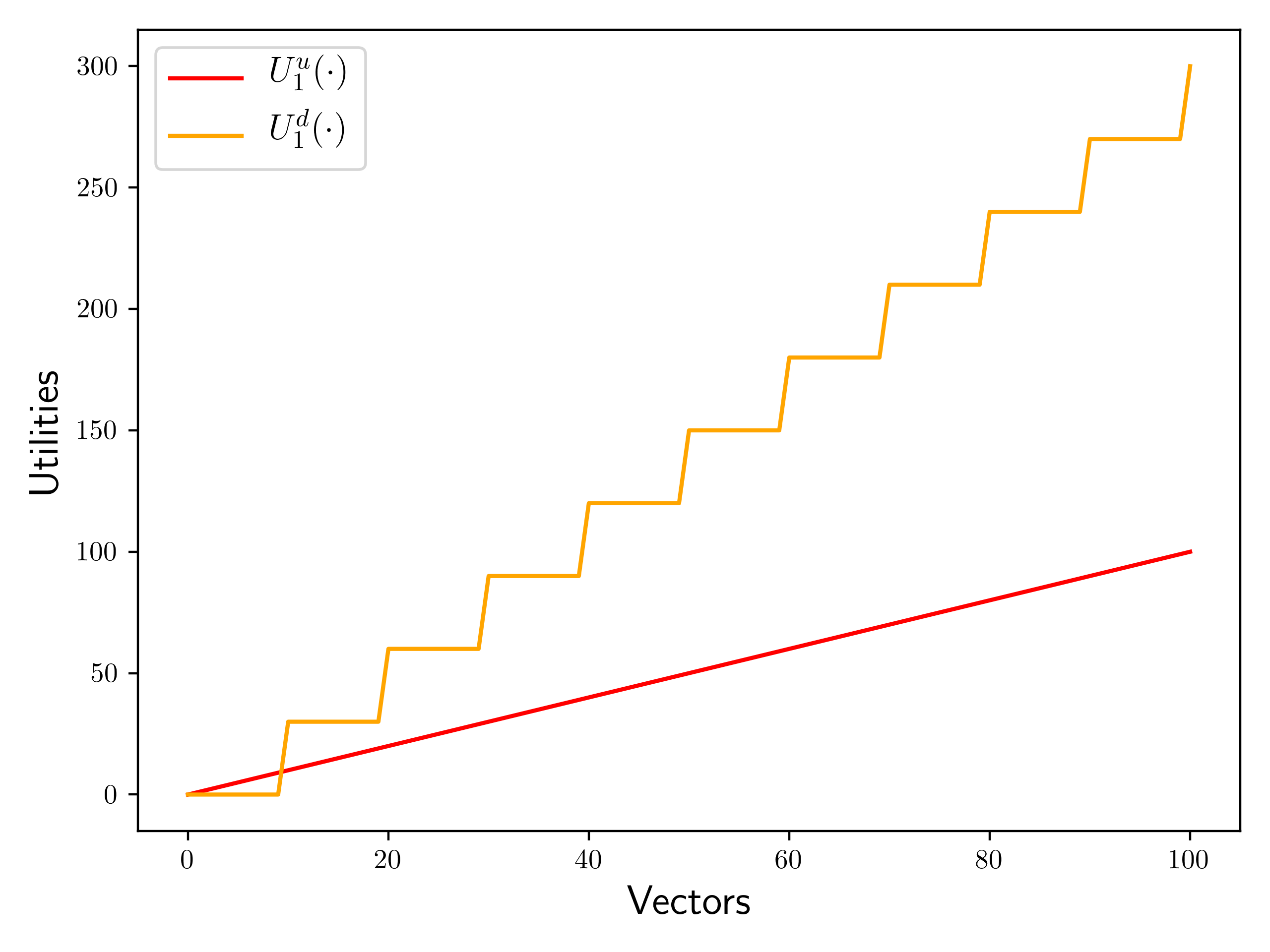}
         \caption{}
         \label{fig.parameters}
     \end{subfigure}
        \caption{Game 1 illustration with only Attacker 1: (a) NE strategies; (b) parameters.}
        \label{fig.complete}
\end{figure}

Figure~\ref{fig.incomplete} illustrates the impact of the presence of more than one type of attacker. In Figure~\ref{fig.BNE.equal} where both attackers are equally likely we observe that, compared to Figure~\ref{fig.BNE}, Attacker 1 benefits from not being the only type of attacker as they play more rewarding vectors than when they were alone because the defender has less interest in detecting him.
However, in Figure~\ref{fig.BNE.non-equal} where the Attacker 2 becomes less likely to appear, the situation of Attacker 1 gets closer to when they were the only attacker and they are reduced to playing less rewarding vectors. On the contrary, Attacker 2 benefits from being less likely as it is less interesting for the defender to detect them so they can play more rewarding vectors.
Note that this situation is much better for Attacker 2 than for Attacker 1. Some of the most rewarding vectors for  Attacker 2 are also used by non-attackers often so they can play them and remain stealthy while the most rewarding vectors of Attacker 1 are almost never used by the non-attackers so they are reduced to playing much less rewarding vectors to remain stealthy.

\begin{figure}
     \centering
     \begin{subfigure}[b]{0.33\textwidth}
         \centering
         \includegraphics[width=\textwidth]{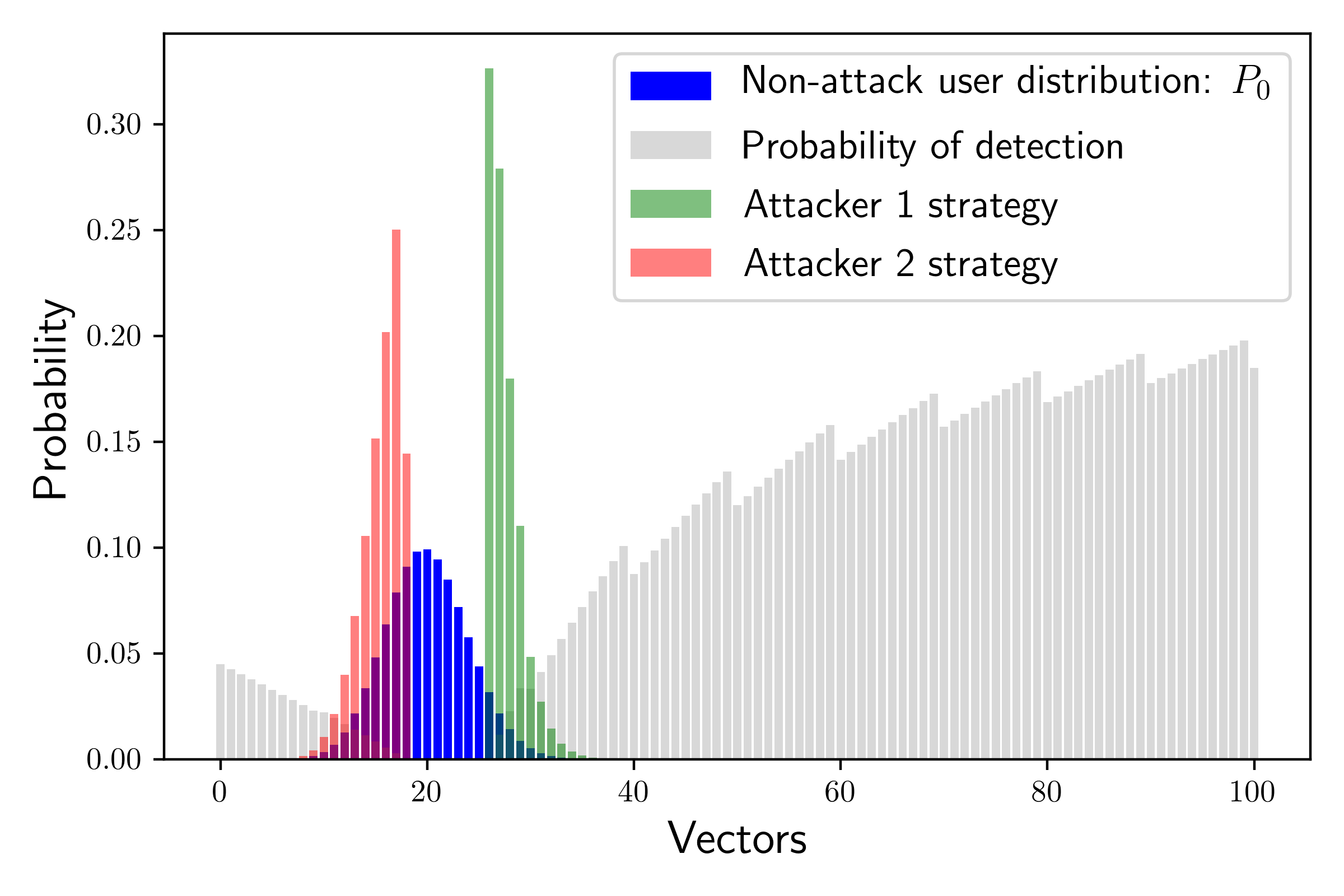}
         \caption{}
         \label{fig.BNE.equal}
     \end{subfigure}
     \begin{subfigure}[b]{0.33\textwidth}
         \centering
         \includegraphics[width=\textwidth]{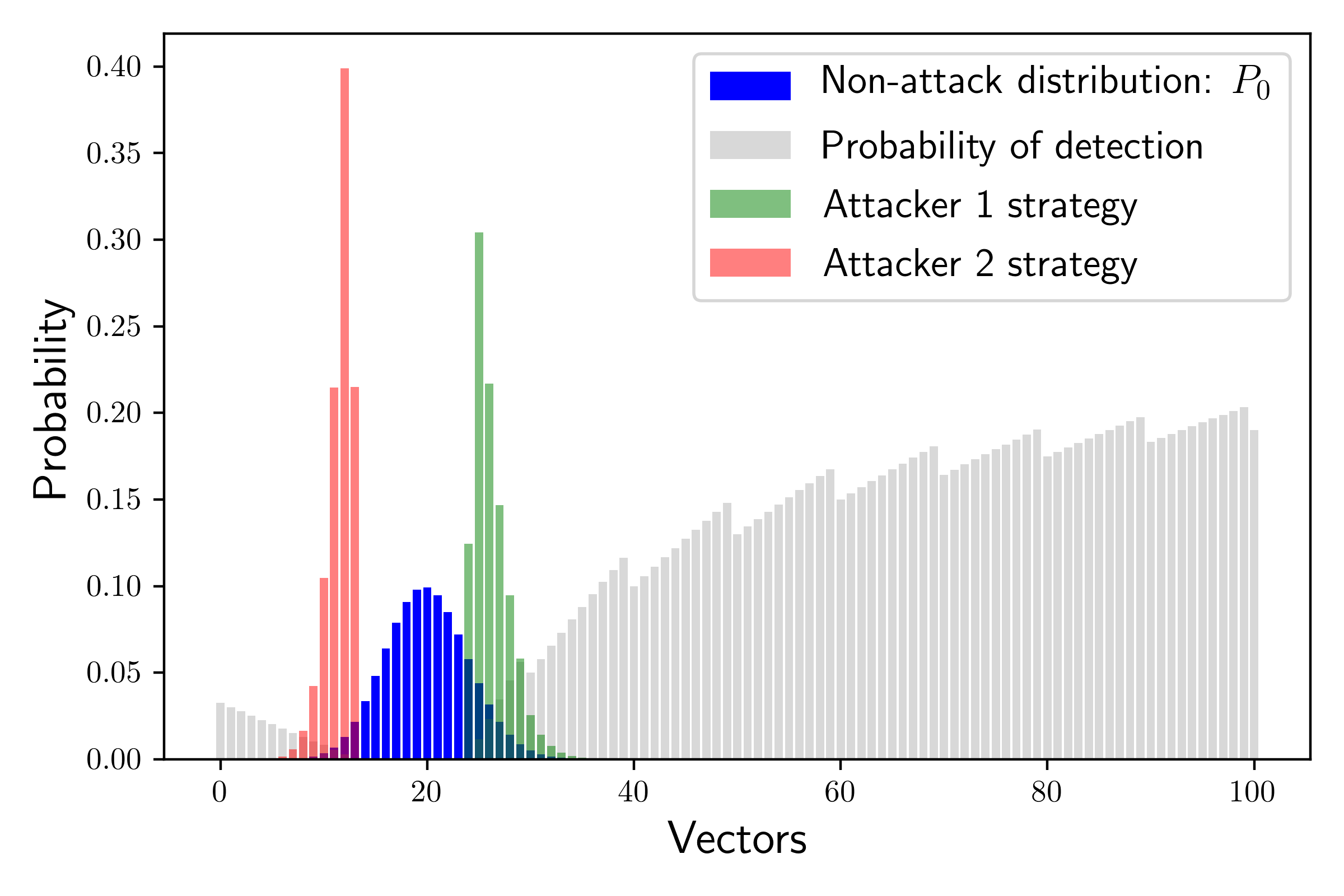}
         \caption{}
         \label{fig.BNE.non-equal}
     \end{subfigure}
        \caption{Game 1 with both attackers: (a) BNE strategies $p_1 = p_1 = 0.5$; (b) BNE strategies $p_1 = 0.95, p_2 = 0.05$.}
        \label{fig.incomplete}
\end{figure}

\subsubsection{Additional illustrations of the training process} 

We now provide illustrations of the training process on Game 2 and Game 3. We use Game 2 to validate our experiments on Section~\ref{sec.experiments} by showing that near-optimal defenses can be found on similar games with small data sets and that large data sets systematically yield a very good approximation. We use Game 3 to illustrate training in a setting with multiple attacker and no simple correspondence between features and costs.

First, Figure~\ref{G_saa} shows the parameter $G$ trained by the defender depending $\ell$ on the setup of Section~\ref{sec.experiments} (using the credit card fraud data set with $N = 284,807$). Recall that $G$ corresponds to the gain of the attacker acting according to its best-response. Thus, a higher $G$ means that the defender is willing to let the attacker gain more. This is compensated by the fact that when $G$ increases, $\pi_G(v)$ decreases and so do the false alarms. We observe that the parameter $G$ increases with the value of $\ell$ which corresponds to the fact that a defender facing higher false alarm costs is less willing to detect non-attacks. We also show computation time for the training on the data set for completeness. These are averaged over $10$ run and plotted with error bars corresponding to one standard deviation and simply show that the training process can be applied with a medium sized data set with reasonable computation times.

\begin{figure}
     \centering

         \centering
\includegraphics[width =0.33\columnwidth]{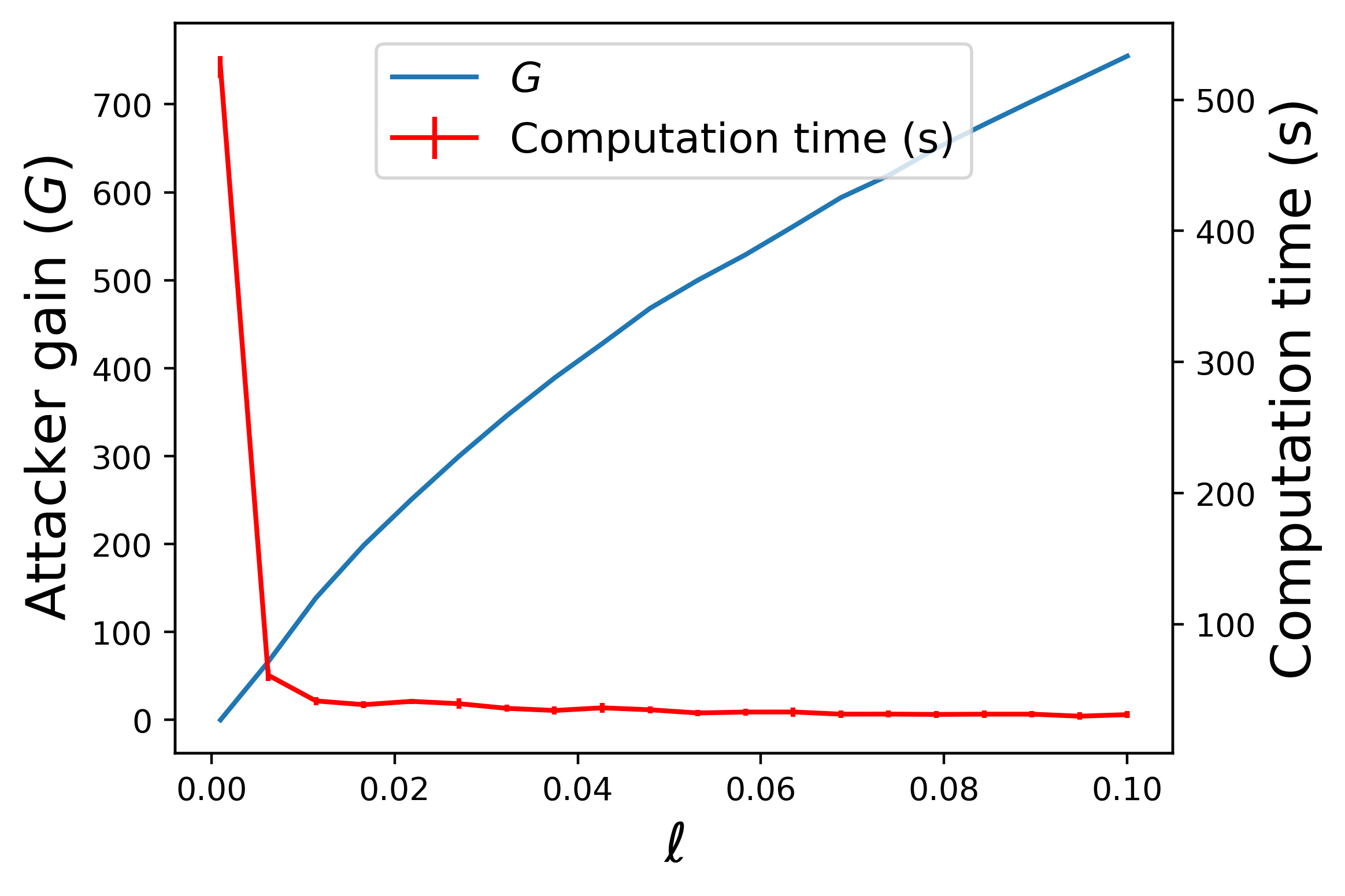}
     \caption{Strategy of the defender training on the data set and computation times}
     \label{G_saa}
\end{figure}

In the following figures, we evaluate the efficiency of our training process using a metric we call approximation ratio which is simply the ratio $100*U^D(\Gbmax)/U^D(\tilde{\Gb})$ (expressed in percentage) where $\Gbmax$ is a maximizer of $U^D$ and $\tilde{\Gb}$ is the outcome of the sample average approximation algorithm. This value is always comprised between $0$ and $100$ with $100$ corresponding to a perfect training. 

Figure~\ref{approx_fraud} illustrates training on Game 2 for different values of the parameter $\ell$ ($0.001$ and the parameters for which we plotted the corresponding probability of detection function in Figure~\ref{empirical}). We plot the ratio between the loss of the defender using the optimal solution and the loss of the defender using the trained solution as well as the probability to obtain the optimal solution through training depending on the size of the data set. These are obtained by running our training algorithm on $300$ different random training set (each training set is generated i.i.d. with replacement). The approximation ratio is the average over the training set and is plotted with error bars corresponding to one standard deviation. The probability to obtain the optimal solution is computed on these random training set. We observe that $\ell = 0.001$ is a best-case for $p_N$. This is caused by the fact that the equilibrium in this setting is trivial. The defender experiences such low false alarm that at equilibrium they classify all vectors as attacks. On the contrary, this is a worse case for the approximation ratio as a slight difference in the proportion of attackers in the training set can lead to a drastic change in strategy in this setting. We also observe that while $p_N$ stays relatively low (but still significant around $0.2$) in all other experiments, the approximation ratio also reaches near $100\%$ on average for data sets of size $5000$. For comparison, we remind that Game 2 is defined with $|\Vset| = 25692$ vectors.
\begin{figure}
\centering
\includegraphics[width = 0.7\columnwidth]{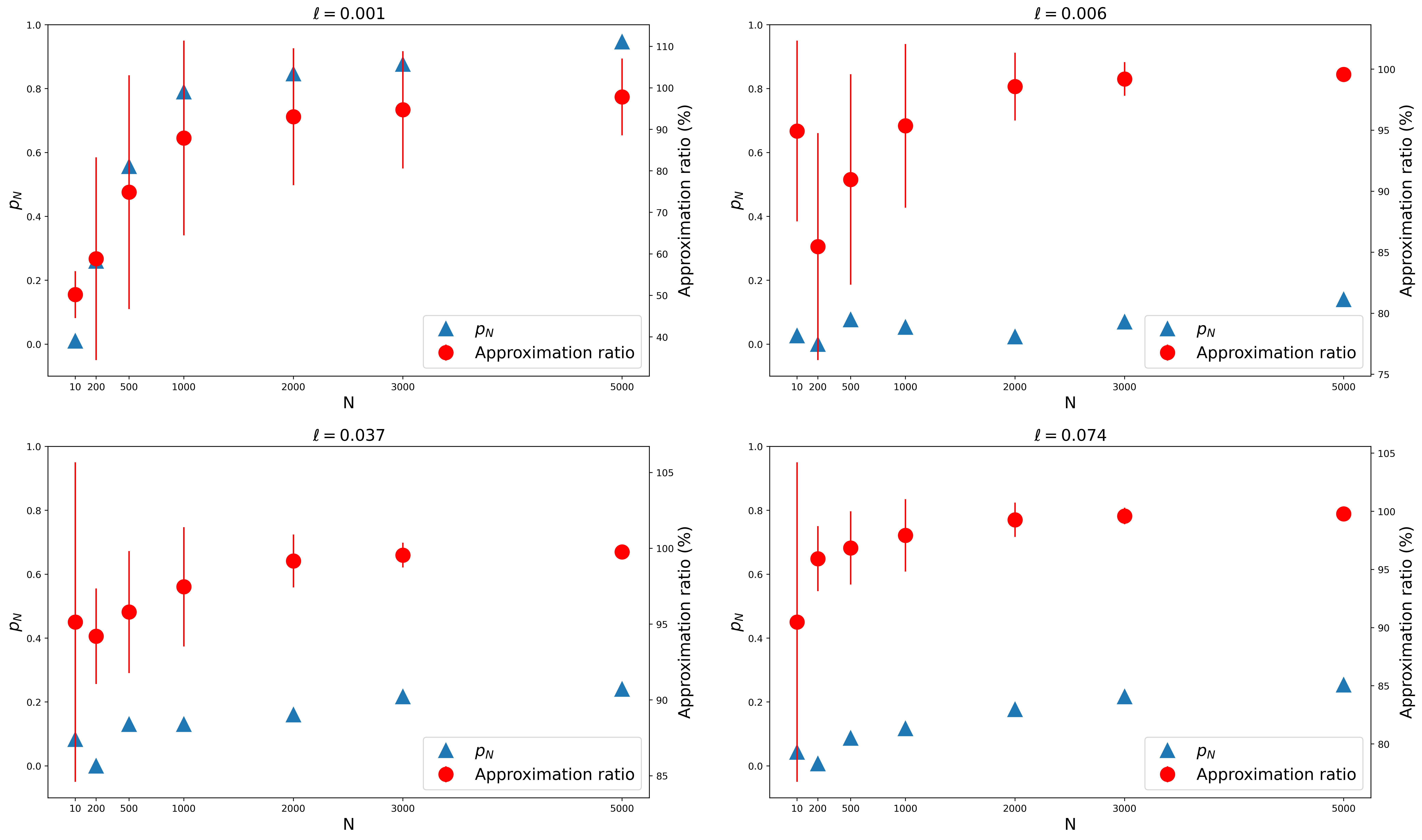}
\caption{Approximation ratio for Game 2.}
\label{approx_fraud}
\end{figure}

Figure~\ref{approx_full} shows the efficiency of the training process through the same metrics when we vary $\ell$ with a training set of the same size as the data set used in Section~\ref{sec.experiments} ($N = 284807$). For each $\ell$ we performed experiments on $10$ random data sets of size $N = 284807$. As previously we plot the averaged approximation ratio with error bars corresponding to one standard deviation. We observe that we systematically obtain the optimal solution, suggesting that a data set of this size is sufficient to correctly learn.
\begin{figure}
\centering
\includegraphics[width = 0.4\columnwidth]{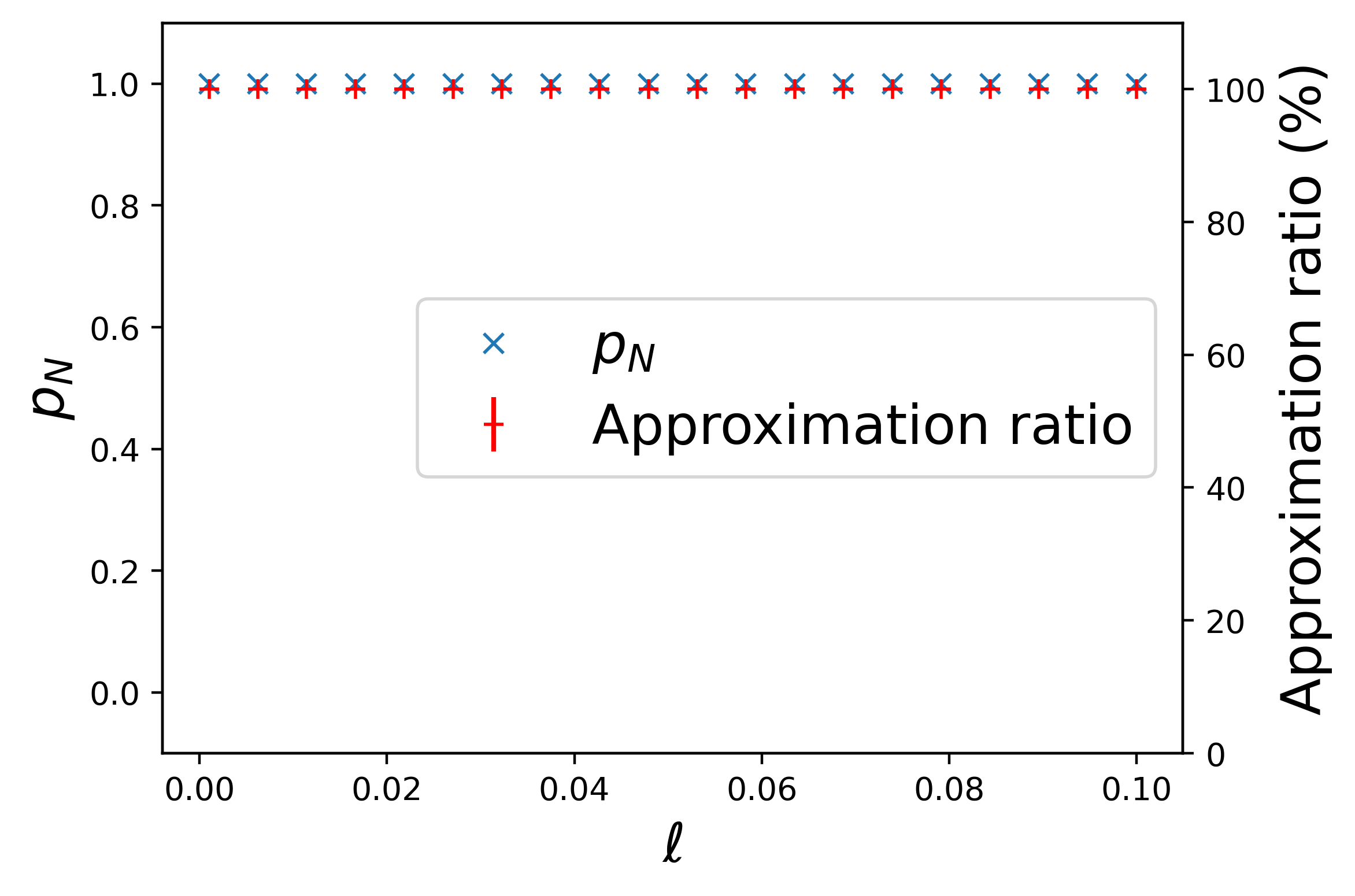}
\caption{Approximation ratio for $N = 284807$}
\label{approx_full}
\end{figure}

Finally, Figure~\ref{approx_random} shows training on Game 3 where for each training set size, we perform experiments $10$ different random game with $4$ types of attackers and $k = 19$. For each game, we perform experiments on $20$ randomly generated training sets (generated iid with replacement) for a total of $200$ experiments per data set size. We observe that the probability to obtain the optimal solution is always null. This is due to the linearity of the problem which makes the optimal values of $\Gb$ being equal to $\LossU(v)$ for some $v$. As there is no relation between the costs of different vectors, with such small data sets, the probability that the vectors corresponding to the optimal parameters are present in the data set are very small and the linear program coming from the sample average approximation procedure cannot find the true optimal. We observe, however, that the approximation ratio is very good even for very small training sets. Note that the efficiency of the approximation is particularly striking in this case as we are able to obtain near-perfect results with training sets of size only $100$ while the number of possible vectors is $|\Vset| = 2^{19}$. This illustrates well the independence in $|\Vset|$ of the training efficiency. Also note that our discussion in Section~\ref{sec.stochastic} about convergence of expected value in at least $N^{-1/2}$ rate translate directly to the approximation ratio. Our graphs suggest, however, that in many of our settings convergence happens at a faster rate.
\begin{figure}
\centering
\includegraphics[width =0.4 \columnwidth]{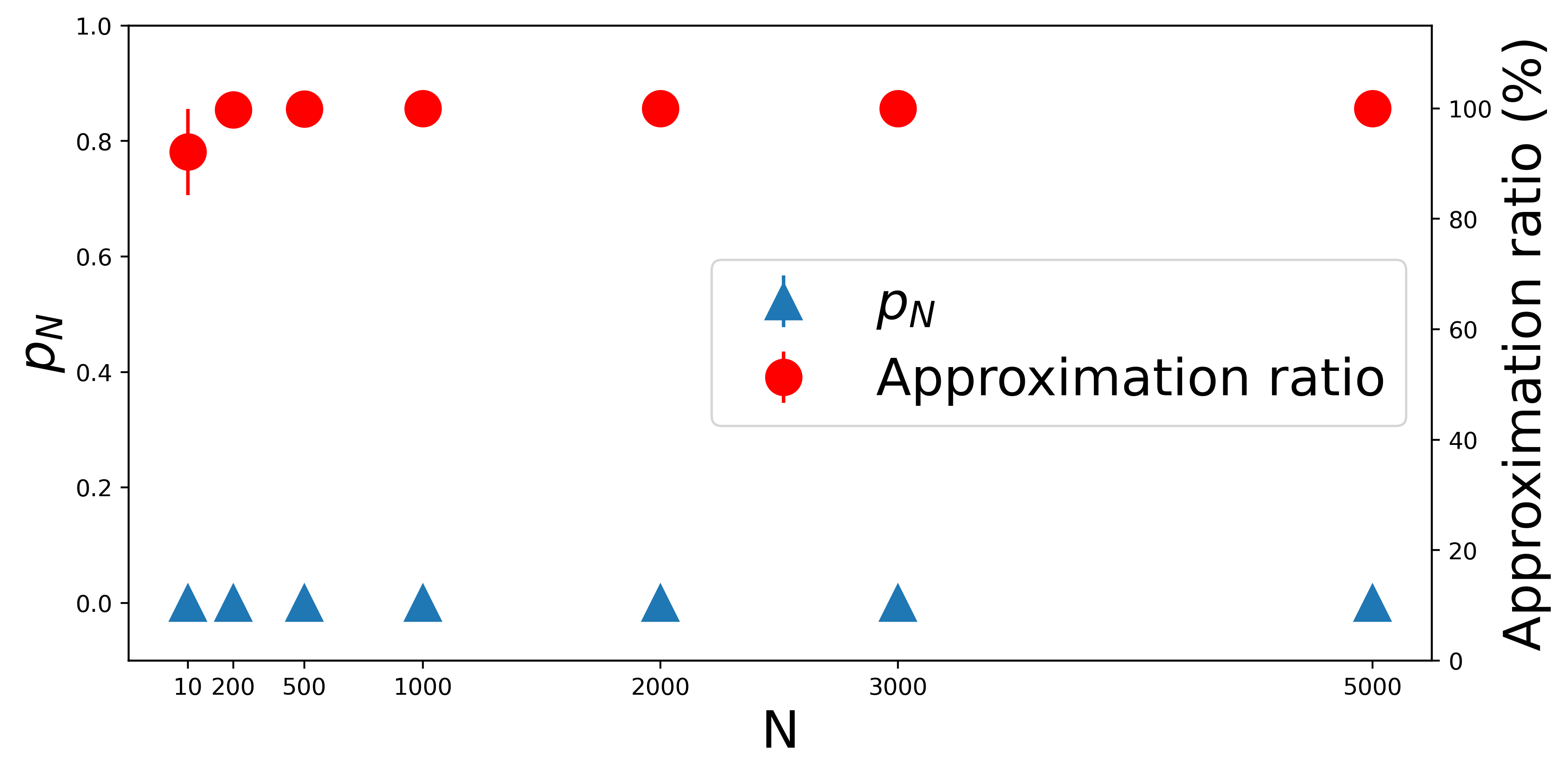}
\caption{Approximation ratio and $p_N$ for Game 3.}
\label{approx_random}
\end{figure}

\subsubsection{Illustrations of the online learning process}

We now illustrate Theorem~\ref{thm.online} using Game 3. We consider the online setting described in Section~\ref{sec.online}:  at each step $t$, the defender faces an attacker of type $i$ with probability $\pa p_i$ and a non-attacker with probability $1 - \pa$. Attackers act according to their best response and non-attacker act according to $\Pna(\cdot)$.

Figure~\ref{online.regret} displays the regret the defender accumulates when learning with Algorithm~\ref{algo.efficient} for different values of $k = \log_2(|\Vset|)$, at time $T=50,000$. For each $k$, we average over $10$ random games. In this experiment, $D$ is close to $40$ in all games which makes the regret bound of Theorem~\ref{thm.online} at least $150,000$; the observed regret is significantly smaller (below $5,000$). We observe that increasing the number of features does not significantly impact the regret, in agreement with the bound of Theorem~\ref{thm.online} that does not depend on the dimensionality of the problem, which illustrates the strength of our parametrization of the defender's strategy.

Figure~\ref{online.distance} displays the distance between the parameters learned by Algorithm~\ref{algo.efficient} and the optimal $\Gbmax$ over time. For this experiment, we run the online algorithm on each game $10$ different times with random starting point for the strategy of the defender. First, we observe that $\Gb_t$ converges towards $\Gbmax$, hence the online strategy converges to the min-max strategy. This is interesting as it is not implied by the no-regret property. Second, it is remarkable to see that the convergence is fast (in $10,000$ steps even when $k=19$). This can seem counterintuitive as one would be unable to learn $\Pna$ in so few steps. However, what we need to learn is only the average false alarm cost associated to a strategy, this is learned fast through the update of the parameters $\Gb$.

\begin{figure}
     \centering
     \begin{subfigure}[b]{0.33\textwidth}
         \centering
       \includegraphics[width=\columnwidth]{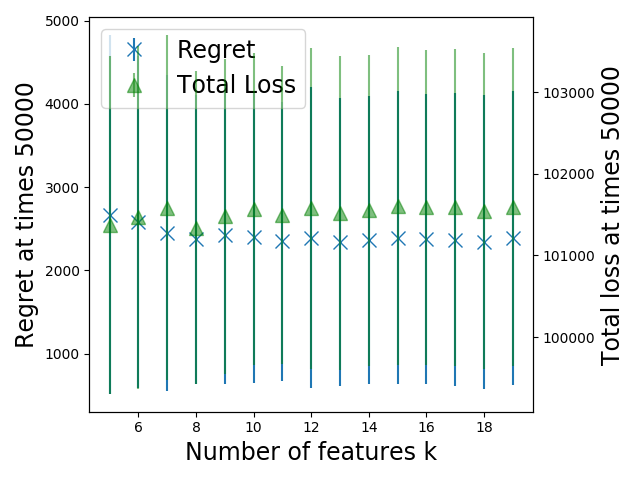}
\caption{Regret for different numbers of features.}
\label{online.regret}
     \end{subfigure}
     \begin{subfigure}[b]{0.33\textwidth}
         \centering
        \includegraphics[width=\columnwidth]{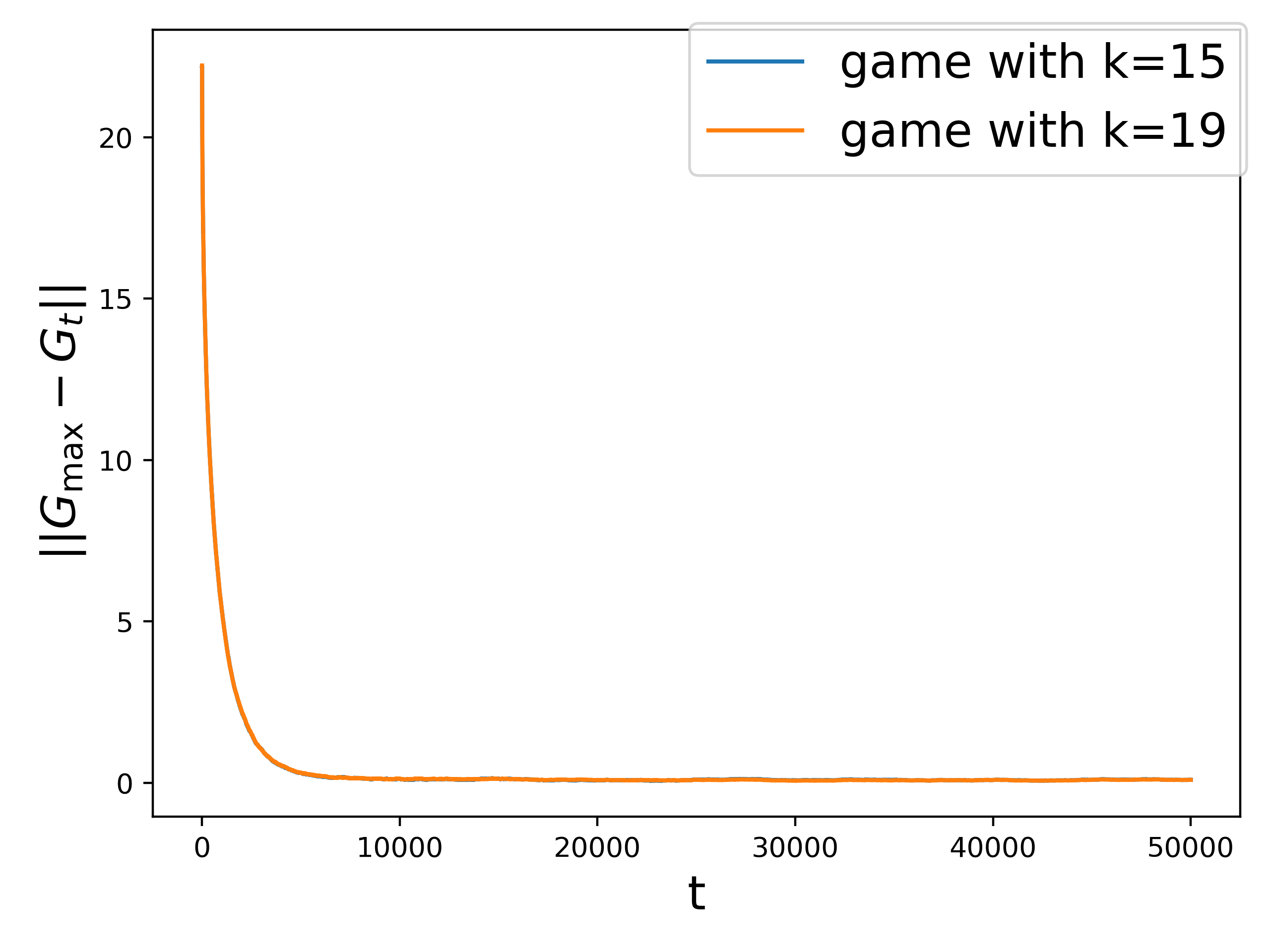}
\caption{Comparison BNE-online strategy.}
\label{online.distance}
     \end{subfigure}
\caption{Regret and distance to the equilibrium}
\end{figure}

Figure~\ref{fig.regret_and_distance} illustrates online learning on Game 3 with $k = 19$ and a number of possible attackers varying from $1$ to $4$. In Figure~\ref{fig.regret_m}, we observe that the regret accumulated by the defender increases with the number of attacker types. This illustrates the fact that our strategy is parametrized by the number of attackers and increasing this number increases the complexity of what we need to learn.
Similarly, Figure~\ref{fig.distance_m} illustrates how far the learned strategy is from the equilibrium during learning.
These two graphs can be contrasted with Figures~\ref{online.regret} and \ref{online.distance} respectively, in which we observed that the regret and distance to equilibrium were not varying with the number of features. 
Thus, this illustrates the fact that our characterization is indeed independent from the number of vectors considered but depends only on the complexity of the characterization, i.e., the number of possible types of attackers.
Finally, Figure~\ref{fig.distance_game_3_err} shows the same plot as Figure~\ref{online.distance} (or Figure~\ref{fig.distance_m} with $m=4$) with error bars that were omitted in the previous plots for readability.

\begin{figure}
     \centering
     \begin{subfigure}[b]{0.33\textwidth}
         \centering
         \includegraphics[width=\textwidth]{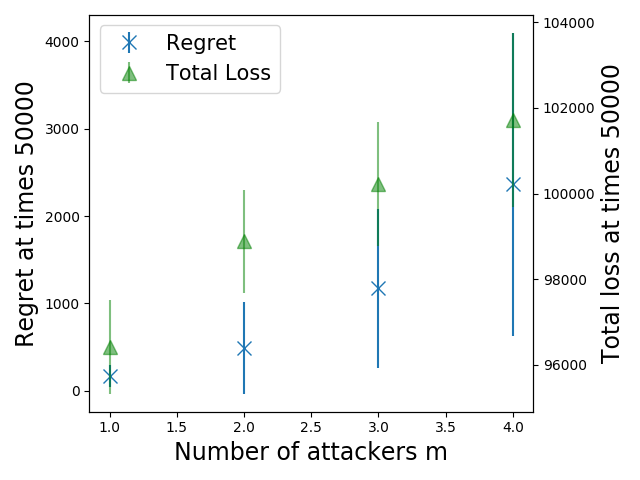}
         \caption{}
         \label{fig.regret_m}
     \end{subfigure}
     \begin{subfigure}[b]{0.33\textwidth}
         \centering
         \includegraphics[width=\textwidth]{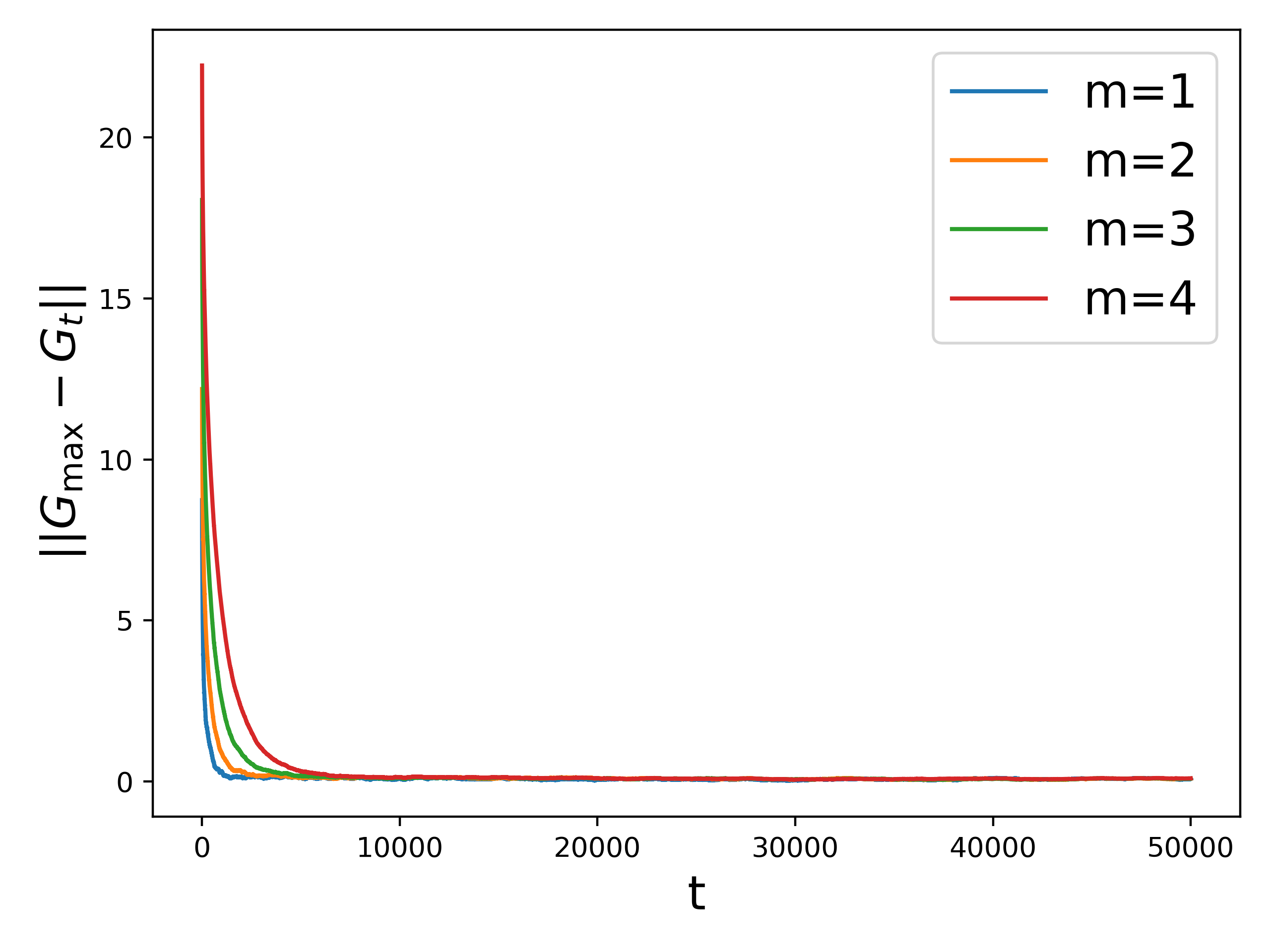}
         \caption{}
         \label{fig.distance_m}
     \end{subfigure}
     \begin{subfigure}[b]{0.33\textwidth}
         \centering \includegraphics[width=\textwidth]{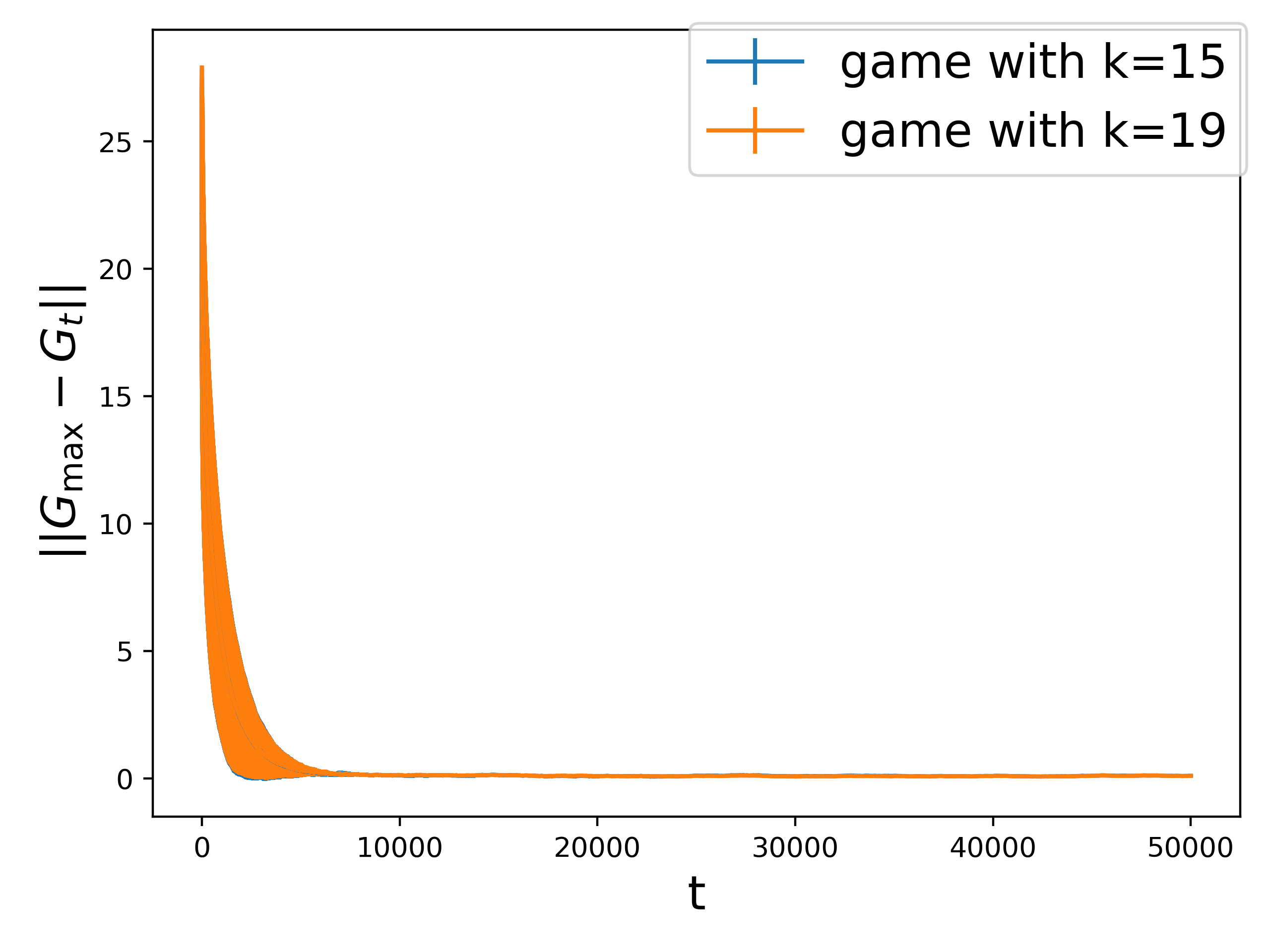}
         \caption{} 
         \label{fig.distance_game_3_err}
     \end{subfigure}
        \caption{Illustration of online learning on Game 2: (a) Regret for different numbers of attackers; (b) Distance to equilibrium for different numbers of attackers; (c) Distance to equilibrium for $m=4$ with error bars.}
        \label{fig.regret_and_distance}
\end{figure}

\end{document}